\documentclass[11pt,letterpaper]{article}

\usepackage[margin=1in]{geometry}        
\usepackage{setspace}                     
\usepackage[T1]{fontenc}                  
\usepackage[utf8]{inputenc}               
\usepackage{lmodern} 
\usepackage{graphicx}
\usepackage{natbib}
\usepackage{authblk}
\usepackage{enumerate}

\usepackage{amsmath,amssymb,amsthm,mathtools,bbm}

\numberwithin{equation}{section}

\theoremstyle{plain}
\newtheorem{theorem}{Theorem}[section]
\newtheorem{lemma}[theorem]{Lemma}
\newtheorem{proposition}[theorem]{Proposition}

\theoremstyle{definition}
\newtheorem{definition}{Definition}
\newtheorem{assumption}{Assumption}

\usepackage[ruled, linesnumbered]{algorithm2e}
\usepackage{xcolor}

\usepackage[
  colorlinks,
  linkcolor=blue,
  citecolor=blue,
  urlcolor=blue,
  breaklinks
]{hyperref}

\def\E{\mathbb{E}}
\def\P{\mathbb{P}}
\def\R{\mathbb{R}}
\def\N{\mathbb{N}}
\newcommand{\bx}{\mathbf{x}}
\newcommand{\bt}{\mathbf{t}}
\newcommand{\bX}{\mathbf{X}}
\newcommand{\bxi}{\boldsymbol{\xi}}
\def\Hcal{\mathcal{H}}
\def\Xcal{\mathcal{X}}

\def\Rcal{\mathcal{R}}
\def\Pcal{\mathcal{P}}

\newcommand{\One}[1]{{\mathbbm{1}}\left\{{#1}\right\}}
\newcommand{\one}[1]{{\mathbbm{1}}_{{#1}}}
\def\conch{{\textsc{CONCH}}}
\def\croc{{\textsc{CROC}}}

\DeclareMathOperator*{\argmax}{argmax}
\DeclareMathOperator*{\argmin}{argmin}

\usepackage[nameinlink,noabbrev]{cleveref}
\crefname{assumption}{assumption}{assumptions}

\newcommand{\papertitle}{Distribution-free root cause analysis}
\newcommand{\paperauthorA}{Rohan Hore}
\newcommand{\paperauthorB}{Aaditya Ramdas}
\newcommand{\affilOne}{Department of Statistics and Data Science, Carnegie Mellon University}

\newcommand{\corrEmail}{rhore@andrew.cmu.edu}

\title{\papertitle}

\author{\paperauthorA\thanks{Corresponding author: \corrEmail}
}
\author{\paperauthorB}
\affil{\affilOne}

\date{\today}

\newcommand{\paperabstract}[1]{%
  \begin{abstract}
    #1
  \end{abstract}
}

\usepackage[nameinlink,noabbrev]{cleveref}

\begin{document}
\maketitle
\paperabstract{
We study distribution-free root cause analysis in multi-stream data, where an evolving underlying system is observed through multiple data streams that may each undergo distributionaleach stream may undergo a changes at an unknown timepoints. In such settings, and the goal is to identify the stream exhibitingwith the earliest change provides a natural starting point for investigating the underlying cause, which we refer to as the \emph{, which we call the root-cause index}. Leveraging conformal $p$-values, we propose a novel framework, Conformal Root Cause Analysis (CROC), which constructs finite-sample valid confidence sets for the root-cause index under minimal assumptions: 
the data streams are independent, and
 within each stream the pre- and post-change observations are sampled exchangeably from arbitrary and unknown distributions. We further establish a universality property, showing that any distribution-free method for root cause localization can be represented within the CROC framework. In addition, under mild regularity conditions and principled score design, our method yields asymptotically sharp confidence sets that efficiently isolate the root cause. We further extend CROC to efficiently handle cross-stream dependence when present. Extensive simulations demonstrate accurate localization of the root stream, supporting our theoretical guarantees.
}

\section{Introduction}
In this paper, we study root-cause analysis in multi-stream data, where observations arise from multiple interacting sources evolving over time and undergoing distributional changes. In many applications, these streams serve as downstream monitors of an underlying system that is not directly observed, and changes in the data reflect latent system-level shifts—for instance,  changes in restaurant operations may be reflected in later variations in customer feedback, while policy changes may result in later shifts in financial indicators. Root-cause analysis aims to identify the mechanism driving such changes \citep{wang2023root,banerjee2009framework,li2025coca}. However, modeling the relationship between observed streams and the underlying system typically requires strong structural or domain assumptions, which may be unreliable in practice. 
In most such settings, temporal ordering provides a natural signal: when changes propagate across many streams, the stream with the earliest change offers a principled starting point for identifying the source of the disruption. Accordingly, in this work, we define the \emph{root-cause index} as the stream whose distribution changes before the others. This serves as an actionable proxy for the underlying cause, particularly in systems where effects propagate downstream over time.

Formally, let $K\in\N$ denote the number of data streams, and write $[m]=\{1,\ldots,m\}$ for any $m\in\N$. For each $k\in[K]$, we observe a sequence of $\mathcal{X}_k$-valued random variables $(X_{1,k},\ldots,X_{n,k})$, where the spaces $\mathcal{X}_k$ may differ across streams (e.g., images in one stream and textual descriptions in another). We assume that each stream may undergo a distributional change at an unknown location. Specifically, there exists an unknown vector $\boldsymbol{\xi} := (\xi_1,\ldots,\xi_K) \in [n]^K$ such that, for each $k\in[K]$,
\[
(X_{1,k},\ldots,X_{\xi_k,k}) \sim \Pcal^{(k)}_{0,\xi_k}, 
\qquad 
(X_{\xi_k+1,k},\ldots,X_{n,k}) \sim \Pcal^{(k)}_{1,\xi_k},
\]
where $\Pcal^{(k)}_{0,\xi_k}$ and $\Pcal^{(k)}_{1,\xi_k}$ denote the pre- and post-change distributions, respectively. Here $\xi_k=n$ indicates that no changepoint occurs in stream $k$. We assume that at least one coordinate of $\boldsymbol{\xi}$ is strictly less than $n$, ensuring the presence of a non-trivial change.

In many applications, additional information is available regarding the joint configuration of changepoints across streams. We encode this via a known constraint set $\Rcal \subseteq [n]^K$ with $\boldsymbol{\xi} \in \Rcal$. The set $\Rcal$ captures structural relationships among the changepoints. For instance, $\Rcal = [n]^K$ allows arbitrary changepoint configurations, while $\Rcal = \{(t,\ldots,t): t\in[n]\}$ corresponds to a fully synchronized setting where all streams share a common changepoint.

\subsection{Distribution-free root cause analysis}

Given $\boldsymbol{\xi}\in \Rcal$, we define the root-cause index as $k_\star :=\arg\min_{k\in[K]} \xi_k$
and assume throughout that the minimizer is unique. Our goal is to identify $k_\star$ in a distribution-free manner. To this end, we do not assume knowledge of the data-generating distribution beyond the following condition.

\begin{assumption}\label{assn:exchangeability_multi}
The data streams are mutually independent across $k\in[K]$. Moreover, for each $k\in[K]$, the pre- and post-change distributions are independent and individually exchangeable, i.e., $\Pcal^{(k)}_{0,\xi_k}\perp \Pcal^{(k)}_{1,\xi_k}$, and for any permutations $\pi:[\xi_k]\to[\xi_k]$ and $\pi':[n]\setminus[\xi_k]\to[n]\setminus[\xi_k]$,
\[
    (X_{1,k},\ldots,X_{\xi_k,k}) \overset{d}{=} (X_{\pi(1),k},\ldots,X_{\pi(\xi_k),k}),\,\,\,
    (X_{\xi_k+1,k},\ldots,X_{n,k}) \overset{d}{=} (X_{\pi'(\xi_k+1),k},\ldots,X_{\pi'(n),k}).
\]
\end{assumption}

We write $\bX := (X_{i,k})_{i\in[n],\,k\in[K]}$ for the full data array. For any $k\in[K]$, let $\Hcal_{0,k}$ denote the hypothesis that $k_\star = k$. We write $\P_k$ and $\E_k$ for probability and expectation under $\Hcal_{0,k}$. We also write $2^{[K]}$ denotes the power set of $[K]$. Under this setup, we define a distribution-free confidence set for the root-cause index as follows.

\begin{definition}
Fix $\alpha\in(0,1)$. A mapping $\mathcal{K}_{1-\alpha} : \mathcal{X}^{n\times K} \to 2^{[K]}$ is called a \emph{distribution-free confidence set for the root-cause index} at level $1-\alpha$ if $\P_{k_\star}\bigl(k_\star \in \mathcal{K}_{1-\alpha}(\bX)\bigr) \;\ge\; 1-\alpha$
for all $\Pcal$ satisfying \Cref{assn:exchangeability_multi}.
\end{definition}

We note that \Cref{assn:exchangeability_multi} is quite mild: it requires only independence across streams and segment-wise exchangeability within each stream, without any further knowledge of the pre- and post-change distributions. This stands in contrast to much of the existing literature on root-cause analysis, which typically relies on structural assumptions or domain knowledge \citep{sole2017survey,zhou2004statistical,alaeddini2011using}. 

A natural approach to this problem is to first localize changepoints in each stream and then subsequently identify the root cause. However, even the simpler task of changepoint localization remains underexplored in a fully distribution-free setting. Broadly, existing approaches of changepoint localization fall into three categories. Parametric methods \citep{kim1989likelihood,saha2026post} often require a separation between pre- and post-change distributions. Nonparametric methods \citep{frick2014multiscale,verzelen2023optimal} typically provide only asymptotic guarantees and may involve difficult-to-calibrate constants. Resampling-based approaches \citep{cho2022bootstrap} are widely used in practice but generally lack rigorous finite-sample guarantees. 

More recently, \citep{dandapanthula2025offline,hore2026conformal} take initial steps toward distribution-free changepoint localization, and these developments serve as key building blocks for our framework.

\subsection{Our contributions}
\label{sec:contributions}

In this paper, we develop a general conformal framework for distribution-free root-cause analysis.
\begin{itemize}
    \item  Building on Conformal Changepoint Localization (\conch{}) \citep{hore2026conformal}, we introduce Conformal Root Cause Analysis (\croc{}), a method that uses conformal $p$-values to construct confidence sets for the root-cause index $k_\star \in [K]$ in multi-stream data.

    \item  We establish exact finite-sample coverage: for any $\alpha\in(0,1)$, the resulting confidence sets contain $k_\star$ with probability at least $1-\alpha$ under minimal exchangeability assumptions (cf.\ \Cref{assn:exchangeability_multi}).

    \item  We prove a universality property: any distribution-free method for root-cause localization can be represented as an instance of \croc{} with an appropriate score function.

    \item  We analyze the role of the score function, characterize an optimal choice, and show that suitable approximations of the optimal score yield asymptotically sharp confidence sets that isolate the true root index.

    \item  We extend \croc{} to handle cross-stream dependence via an aggregation-based construction. lastly, we provide simulation studies demonstrating the effectiveness of \croc{} across a range of settings, showing accurate and sharp localization of the root-cause index in practice.
\end{itemize}

\paragraph{Organization of the paper}
The rest of the paper is organized as follows. In Section~\ref{sec:method}, we introduce the \croc{} framework. Section~\ref{sec:validity} establishes its finite-sample validity, and Section~\ref{sec:universality} proves its universality. In Section~\ref{sec:CPPscore}, we study the asymptotic sharpness of the resulting confidence sets. Section~\ref{sec:croc_dependence} extends the framework to handle cross-stream dependence. Section~\ref{sec:experiments} presents simulation experiments, and Section~\ref{sec:discussion} concludes with a brief discussion.
\section{Methodology}\label{sec:method}

In this section, we develop a general framework for distribution-free root cause analysis, based on conformal $p$-value machinery. Since introduction in \cite{vovk1999machine}, conformal $p$-values have been extensively employed in several other applications \citep{bates2023testing,wu2024conditional,jin2023selection}. Apart from them, the conformal changepoint localization (\conch{}), introduced in \cite{hore2026conformal}, provides a distribution-free changepoint localization in the single-stream setting. Motivated by the \conch{} procedure, we introduce a novel method for root cause analysis, which we refer to as \emph{Conformal Root Cause Analysis} (CROC).

\subsection{Background: changepoint localization in single data-stream}\label{sec:conch_single}

We begin by reviewing the \conch{} procedure, which provides a distribution-free confidence set for the true changepoint when there is a single data stream. Beyond serving as background, the purpose of this section is to also highlight the key principle that will guide the construction of \croc.

To that end, consider the case $K=1$, where we observe a sequence $\bX:=(X_1,\ldots,X_n)\in \mathcal{X}^n$ with an unknown changepoint $\xi\in[n-1]$, such that \Cref{assn:exchangeability_multi} holds.
The implementation of \conch{} requires specifying two key components:

\begin{enumerate}[(i)]
    \item \textbf{ChangePoint Plausibility (CPP) score.} This is any mapping $S:\mathcal{X}^n \to \R^{n-1}$. While there is no restriction on the choice of $S$ to implement the algorithm, the guiding philosophy is that the $t$-th coordinate $S_t$ should quantify the plausibility of $t$ being a changepoint, with larger values indicating stronger evidence.

    \item \textbf{Split-permutation group.} For each $t\in[n-1]$, define
    \[
    \Pi_t := \Bigl\{ \pi \in \mathcal{S}_n : \pi(i) \le t \ \text{for all } i \le t,\ \ \pi(i) > t \ \text{for all } i > t \Bigr\},
    \]
    i.e., the set of permutations that act separately on the pre- and post-change segments.
\end{enumerate}

Let $\Hcal'_{0,t}$ denote the hypothesis that $\xi=t$. Under \Cref{assn:exchangeability_multi}, if $t$ is the true changepoint, then permutations in $\Pi_t$ preserve the pre- and post-change exchangeability, and hence leave the distribution of $S_t(\bX)$ invariant under $\Hcal'_{0,t}$. 
This invariance then leads to the conformal $p$-value
\begin{equation}\label{eq:pvalue_conch}
    p_t := \frac{1}{|\Pi_t|}\sum_{\pi\in\Pi_t} \One{ S_t(\pi(\bX)) \le S_t(\bX)},
\end{equation}
which is super-uniform under $\Hcal'_{0,t}$. Finally, the corresponding confidence set is obtained by thresholding these $p$-values appropriately, and defining $\mathcal{C}^{\mathrm{CONCH}}_{1-\alpha} := \{t\in[n-1]: p_t > \alpha\}$.

This construction is not only distribution-free valid, but also admits a natural interpretation: building a confidence set for $\xi$ is equivalent to testing the hypotheses $\Hcal'_{0,t}$ for all candidate changepoint locations $t\in[n-1]$. 

\subsection{Conformal root cause analysis}\label{sec:croc_defn}

\IncMargin{1.2em} 
\begin{algorithm}[t]
    \caption{\croc{}: conformal root cause analysis}
    \label{alg:croc}
    \KwIn{$\bX = (X_{i,k})_{i\in[n],\,k\in[K]}$ (data), $1-\alpha$ (target coverage), $S$ (CPP score), $\Rcal$ (constraint set)}
    \KwOut{$\mathcal{K}^{\mathrm{CROC}}_{1-\alpha}$ (confidence set for root-cause index)}

    \For{$\bt \in \Rcal$}{
        $\Pi_{\bt} \gets \{(\pi_1,\ldots,\pi_K): \pi_k \text{ permutes within } [1,t_k] \text{ and } [t_k+1,n]\}$\;
        \ForEach{$\pi \in \Pi_{\bt}$}{Evaluate $S(\pi(\bX),\bt)$\;}
        $p_{\bt} \gets \frac{1}{|\Pi_{\bt}|}\sum_{\pi\in \Pi_{\bt}} \One{S(\pi(\bX),\bt)\le S(\bX,\bt)}$\;
    }

    \For{$k \in [K]$}{
        $I_k \gets \{\bt \in \Rcal : t_k < t_j \text{ for all } j\neq k\}$\;
        $p_{(k)} \gets \max_{\bt \in I_k} p_{\bt}$\;
    }

    $\mathcal{K}^{\mathrm{CROC}}_{1-\alpha} \gets \{k \in [K] : p_{(k)} > \alpha\}$\;
    \Return{$\mathcal{K}^{\mathrm{CROC}}_{1-\alpha}$}
\end{algorithm}
\DecMargin{1.2em}

We now extend the correspondence between hypothesis testing and confidence sets to the multi-stream setting for root-cause analysis. Building on this philosophy, for this task, it suffices to test the null hypotheses $\Hcal_{0,k}$ for each $k\in [K]$.

Under the null $\Hcal_{0,k}$, the true changepoint configuration $\bxi$ must lie in the set $I_k$, where we define
\[
I_k := \{(t_1,\ldots,t_K)\in \Rcal : t_k < t_j \ \text{for all } j\neq k\}.
\]
Equivalently, we may write $\Hcal_{0,k} = \bigcup_{\bt\in I_k} \Hcal'_{0,\bt}$,
where $\Hcal'_{0,\bt}$ denotes the hypothesis that the changepoint configuration $\bxi$ equals $\bt$. This representation suggests that it suffices to first construct distribution-free tests for $\Hcal'_{0,\bt}$ for each $\bt\in \Rcal$.

To this end, we extend the \conch{} framework to the multi-stream setting, retaining the same terminology for its key components.
\begin{enumerate}[(i)]
    \item \textbf{Changepoint plausibility (CPP) score:} any mapping $S:(\prod_{k=1}^K \Xcal_k^n)\times \Rcal \to \R$,
    where $S(\bX,\bt)$ is chosen such that it quantifies the plausibility of $\bt$ being the true changepoint vector.

    \item \textbf{Split-permutation group:} for each $\bt\in\Rcal$, we define $\Pi_{\bt}$ as the set of permutations $\pi=(\pi_1,\ldots,\pi_K)\in (\mathcal{S}_n)^K$ such that, for each $k\in[K]$,
    \[
    \pi_k(i)\le t_k \ \text{for } i\le t_k,~~\text{and} \quad 
    \pi_k(i)>t_k \ \text{for } i>t_k.
    \]
    That is, $k$-th stream is permuted independently within the segments to the left and right of $t_k$.
\end{enumerate}

If $\bt$ is the true changepoint vector, then permutations in $\Pi_{\bt}$ preserve the joint distribution of $\bX$ under $\Hcal'_{0,\bt}$, which yields the conformal $p$-value
\begin{equation}\label{eq:pval_croc}
p_{\bt} := \frac{1}{|\Pi_{\bt}|} \sum_{\pi\in \Pi_{\bt}} \One{ S(\pi(\bX), \bt) \le S(\bX,\bt)}.
\end{equation}
Finally, to test $\Hcal_{0,k}$, we aggregate these $p$-values over all $\bt\in I_k$ via max-aggregation:
\begin{equation}\label{eq:pval_croc_agg}
p_{(k)} := \max_{\bt\in I_k} p_{\bt}.
\end{equation}
We then construct the \croc{} confidence set by thresholding these aggregated $p$-values:
\[
\mathcal{K}^{\mathrm{CROC}}_{1-\alpha} := \{k\in[K]: p_{(k)} > \alpha\}.
\]
The resulting procedure is summarized in \Cref{alg:croc}.
In the next section, we formally show that $\mathcal{K}^{\mathrm{CROC}}_{1-\alpha}$ is a valid confidence set for $k_\star$.

\section{Finite-sample validity of \croc{}}\label{sec:validity}

In this section, we establish the finite-sample coverage of $\mathcal{K}^{\mathrm{CROC}}_{1-\alpha}$. Recall that $\Hcal'_{0,\bt}$ denotes distributions satisfying \Cref{assn:exchangeability_multi} with changepoint configuration $\bt$, while $\Hcal_{0,k}$ denotes those satisfying \Cref{assn:exchangeability_multi} with root index $k$.

\begin{theorem}[Validity]\label{thm:validity_croc}
Fix $\alpha\in(0,1)$. For each $\bt\in \Rcal$, the $p$-value $p_{\bt}$ in \eqref{eq:pval_croc} satisfies $\mathbb{P}_{\Hcal'_{0,\bt}}(p_{\bt} \le \alpha) \le \alpha$. Thus, $\mathbb{P}_{k}(p_{(k)} \le \alpha) \le \alpha$, and $\mathcal{K}^{\mathrm{CROC}}_{1-\alpha}$ is a distribution-free confidence set for $k_\star$.
\end{theorem}

This is a distribution-free result, meaning that the finite-sample validity of $\mathcal{K}^{\mathrm{CROC}}_{1-\alpha}$ holds for any distribution $\Pcal$ satisfying \Cref{assn:exchangeability_multi}. Further, this validity holds for any CPP score function $S$, and thus it offers great flexibility in the choice of CPP score for \croc; in particular, one may learn the score from the observed data $\bX$ itself. 

Note that the $p$-values defined in~\eqref{eq:pval_croc} require evaluation of $S(\cdot,\bt)$ for all permutations $\pi\in \Pi_{\bt}$. However, as the sample size in each stream, $n$, increases, even for a moderate number of streams, $K$, the size of $\Pi_{\bt}$ becomes large, and computing the $p$-value becomes expensive. A natural remedy is to use a Monte Carlo approximation, i.e., sample $\pi^{(1)},\ldots,\pi^{(M)} \stackrel{\text{i.i.d.}}{\sim} \text{Unif}(\Pi_{\bt})$, and then calculate:
\begin{equation}\label{eq:pvalue_croc_MC}
    \tilde{p}_{\bt} := \frac{1 + \sum_{m=1}^M \One{ S(\pi^{(m)}(\bX),\bt) \leq S(\bX,\bt)}}{1+M}.
\end{equation}  
We then define $\tilde{p}_{(k)}:= \max_{\bt\in I_k} \tilde{p}_{\bt}$. Importantly, for any permutation in $\Pi_{\bt}$, under the null $\Hcal'_{0,\bt}$, the joint distribution of $\bX$ is preserved. Hence, even with this randomized construction, $\tilde{p}_{\bt}$ is a valid $p$-value under $\Hcal'_{0,\bt}$. Consequently, the randomized confidence set $\tilde{\mathcal{K}}^{\mathrm{CROC}}_{1-\alpha} := \{k\in[K]: \tilde{p}_{(k)} > \alpha\}$ is a valid confidence set for $k_\star$. We record this in the following theorem.

\begin{theorem}\label{thm:validity_croc_randz}
Fix $\alpha\in(0,1)$. For each $\bt\in \Rcal$, $\mathbb{P}_{\Hcal'_{0,\bt}}(\tilde{p}_{\bt} \le \alpha) \le \alpha$. Thus, $\mathbb{P}_{k}(\tilde{p}_{(k)} \le \alpha) \le \alpha$, and $\tilde{\mathcal{K}}^{\mathrm{CROC}}_{1-\alpha}$ is a distribution-free confidence set for $k_\star$.
\end{theorem}
\section{Universality of the \croc{} algorithm}\label{sec:universality}

While, through the choice of CPP score $S$, the \croc{} framework captures a broad class of methods for constructing distribution-free confidence sets for $k_\star$, a natural question is whether any alternative distribution-free procedures for root-cause analysis may exist. In this section, we give a definitive answer: any distribution-free procedure for root-cause analysis must be an instance of the \croc{} framework with an appropriate choice of CPP score.

\begin{theorem}[Universality of \croc]
\label{thm:universality_croc}
Fix $\alpha\in(0,1)$. Let $C$ be any procedure that maps $\bX \in \prod_{k=1}^K \Xcal_k^n$ to a set $C(\bX) \subseteq \Rcal$ such that $\P_{k}(k\in C(\bX))\geq 1-\alpha$. Then $C$ coincides exactly with the set $\mathcal{K}^{\mathrm{CROC}}_{1-\alpha}$ constructed using the CPP score $S(\bx,\mathbf{t}) = \mathbbm{1}\bigl\{\argmin_{k\in [K]}t_k\in C(\bx)\bigr\}$.
\end{theorem}

Therefore, \croc{} captures the entire class of distribution-free root-cause analysis methods. Beyond its conceptual significance, it naturally provides a principled way to calibrate any existing root-cause localization method that may be only valid asymptotically or under stringent model assumptions. As mentioned before, a common approach is to first perform changepoint localization in each data stream, and then develop a method for root-cause analysis. Such approaches, often built under parametric model assumptions, are very sensitive to model misspecification. Therefore, we may wrap such existing methods within the \croc{} framework via a CPP score designed from these localization methods, and restore validity. If the original procedure $C$ is already distribution-free valid, this calibration step leaves it unchanged and preserves its structural properties; otherwise, it refines $C$ into a distribution-free valid confidence set.
\section{Role of CPP score}\label{sec:CPPscore}

While validity holds for any choice of CPP score, an equally important consideration in practice is the sharpness of the resulting \croc{} confidence sets. In particular, appropriately tailored score functions can lead to substantially narrower confidence sets. To start with, the following proposition sets up some basic guidelines regarding the choice od $S$.

\begin{proposition}\label{prop:score-properties} Fix $n\in\mathbb{N}$ and $\alpha\in(0,1)$. 
\begin{enumerate}[(i)] 
\item \textnormal{\textbf{(Symmetry yields power loss).}} Fix $t\in \Rcal$. If $S$ satisfies $S(\cdot,\bt)=S(\pi(\cdot),\bt)$ for all $\pi\in \Pi_t$, then the $p$-value $p_t$ in \eqref{eq:pval_croc} equals $1$, and thus, $\mathbb{P}(\argmin_{k\in [K]}t_k \in \mathcal{K}^{\mathrm{CROC}}_{1-\alpha})=1$. 
\item \textbf{(Conformal data-processing inequality).} Let $C_1$ be the \croc{} set based on score $S$. For any non-decreasing $f:\mathbb{R}\to\mathbb{R}$, let $C_2$ be the corresponding set based on $f(S)$. Then $C_1\subseteq C_2$.
\end{enumerate} \end{proposition}
To put in words, the first part suggests that using the observed data $\bX$ in symmetric way yields trivial $p$-values and hence leads to large \croc{} confidence sets. Such scores should be avoided in practice. Second part is a monotonicity property of \croc{}: applying any non-decreasing transformation to the CPP score can only enlarge the resulting set, and any \emph{strictly} increasing transformation leaves it unchanged. 

\subsection{Optimal CPP score}
In order to characterize an expression for the optimal choice of score, we consider a canonical setting, namely the i.i.d.\ changepoint model. Suppose $\Pcal^{(k)}_{0,\xi_k}=(\Pcal_{0,k})^{\xi_k}$ and $\Pcal^{(k)}_{1,\xi_k}=(\Pcal_{1,k})^{n-\xi_k}$, where $\Pcal_{0,k}$ and $\Pcal_{1,k}$ admit densities $f_{0,k}$ and $f_{1,k}$ with respect to a common dominating measure $\mu$. That is, within each stream $k$, pre-change observations are i.i.d.\ from $f_{0,k}$ and post-change observations are i.i.d.\ from $f_{1,k}$. Let $\Pcal_{\mathrm{IID}}$ denote the class of distributions consistent with this model and satisfying \Cref{assn:exchangeability_multi}. 

Now, we observe that the task of building a narrow \croc{} confidence set is tied to the task of constructing tests that have greater power against the hypothesis $\{\Hcal_{0,k}\}_{k\neq k_\star}$, or equivalently constructing tests that have greater power against the hypothesis $\Hcal'_{0,\bt}$ for any $\bt\neq \bxi$.
Hence, the optimal choice of CPP score can be viewed through the lens of testing the null $\Hcal_{0,\boldsymbol{\xi}}$ against alternative $\Hcal_{0,\mathbf{t}}$, with the goal of maximizing power. That said, within the framework of \croc{}, we are still limited to this testing problem only via conformal $p$-values. However, a non-trivial formulation of the testing problem conditional on suitable multisets, and by applying the classical Neyman--Pearson Lemma, we get the following result.

\begin{theorem}\label{thm:optimal_score}
Any strictly increasing transformation of the CPP score $S^{\mathrm{OPT}}$ defined as
\begin{equation}\label{eq:optimal_CPP_score}
S^{\mathrm{OPT}}(\bx,\mathbf{t})
=\sum_{k=1}^K \log\left(
\frac{\prod_{i\le t_k} f_{0,k}(x_{i,k})\prod_{i> t_k} f_{1,k}(x_{i,k})}
{\prod_{i\le \xi_k} f_{0,k}(x_{i,k})\prod_{i> \xi_k} f_{1,k}(x_{i,k})}
\right).
\end{equation}
gives the optimal 
minimizes the expected size of the \croc{} confidence set. In particular, for any $k\in [K]$ and any CPP score $S:(\prod_{k=1}^K \Xcal_k^n)\times \Rcal \to \R$, it holds that $\E_{\Hcal_{0,k}\,\cap\,\Pcal_{\mathrm{IID}}}\!\big[\,|\mathcal{K}^{\mathrm{CROC}}_{1-\alpha}(S)|\,\big]
\;\ge\;
\E_{\Hcal_{0,k}\,\cap\,\Pcal_{\mathrm{IID}}}\!\big[\,|\mathcal{K}^{\mathrm{CROC}}_{1-\alpha}(S^{\mathrm{OPT}})|\,\big]$.
\end{theorem}

\subsection{Practical CPP scores}

The optimal score $S^{\mathrm{OPT}}$ depends on the unknown quantities $\boldsymbol{\xi}$ and $\{f_{0,k},f_{1,k}\}_{k=1}^K$, and is therefore not directly implementable. In practice, we construct a data-driven approximation. Specifically, we first obtain estimates $\{\hat{f}_{0,k},\hat{f}_{1,k}\}_{k=1}^K$, and then replace $\boldsymbol{\xi}$ by a maximum likelihood estimate $\hat{\boldsymbol{\xi}}=(\hat{\xi}_1,\ldots,\hat{\xi}_K)$, where
\begin{equation}\label{eq:mle}
    \hat{\xi}_k := \arg\max_{t\in[n]} \prod_{i\le t} \hat{f}_{0,k}(x_{i,k}) \prod_{i>t} \hat{f}_{1,k}(x_{i,k}).
\end{equation}
We then define the CPP score, referred to as the \emph{learned CPP score},
\begin{equation}\label{eq:practical_CPP_score}
S(\bx,\mathbf{t})
=\prod_{k=1}^K \left(
\frac{\prod_{i\le t_k} \hat{f}_{0,k}(x_{i,k})\prod_{i> t_k} \hat{f}_{1,k}(x_{i,k})}
{\prod_{i\le \hat{\xi}_k} \hat{f}_{0,k}(x_{i,k})\prod_{i> \hat{\xi}_k} \hat{f}_{1,k}(x_{i,k})}
\right).
\end{equation}

In some settings, the true pre- and post-change densities $\{f_{0,k},f_{1,k}\}_{k=1}^K$ may be known. In this case, we take $\hat{f}_{0,k}=f_{0,k}$ and $\hat{f}_{1,k}=f_{1,k}$ for each $k\in [K]$, and refer to the corresponding version of~\eqref{eq:practical_CPP_score} as the \emph{oracle CPP score}.

More generally, the densities can be learned from $\bX$ independently within each data stream. For instance, for the $k$-th stream, we may treat the observations $\{X_{1,k},\ldots,X_{n,k}\}$ as an unordered collection, apply an unsupervised clustering method to partition the data into two groups, and then use these groups to estimate the pre- and post-change densities separately.

In particular, in high-dimensional settings where direct density estimation is challenging, one may instead train a classifier to distinguish between pre- and post-change samples, and use the resulting logits as a proxy for the log-likelihood ratio. These logits can then be plugged directly into~\eqref{eq:practical_CPP_score}, yielding a computationally efficient and scalable CPP score. This approach is especially useful in settings such as high-dimensional text or image data.

\subsection{Asymptotic sharpness of \croc{} confidence sets}

We now show that, under mild conditions, the \croc{} confidence sets concentrate around the true root-cause index as the sample size grows.

Consider a sequence of root-cause analysis problems with $K_n \to \infty$ and $K_n/n \to 0$ as $n\to \infty$. For each $n$, we observe $\bX = \bigl((X_{i,j,n})\bigr)_{i\in[n],\,j\in[K_n]}$, where stream $j$ has a changepoint at $\xi_{j,n}\in[n-1]$ and follows an i.i.d.\ changepoint model with pre- and post-change distributions $P_{0,j}$ and $P_{1,j}$ admitting densities $f_{0,j},f_{1,j}$. Let $(p_{(k),n})_{k\in [K]}$ denote the resulting \croc{} $p$-values and $\mathcal{K}^{\croc}_{n,1-\alpha}$ the corresponding confidence set.

We analyze the oracle log-likelihood ratio (LLR) score $\ell_j(x)=\log\!\bigl(f_{0,j}(x)/f_{1,j}(x)\bigr)$, and impose a set of mild and interpretable conditions. First, we assume non-degenerate changepoints, in the sense that they lie away from the boundaries, i.e., $\min_{j}\min(\xi_{j,n},\,n-\xi_{j,n})/n \ge \tau$ for some $\tau\in(0,1/2)$. Second, we assume that in each stream, the pre- and post-change distributions are distinguishable, $\mathrm{KL}(P_{0,j}\|P_{1,j}),\,\mathrm{KL}(P_{1,j}\|P_{0,j})>0$, and the LLRs have uniformly bounded variance, $\sup_j \max\{\mathrm{Var}_{P_{0,j}}(\ell_j),\,\mathrm{Var}_{P_{1,j}}(\ell_j)\}<\infty$. Finally, we require a separation condition ensuring that the root changepoint appears sufficiently earlier than the others, namely $\min_{j\neq k_\star} (\xi_{n,j}-\xi_{n,k_\star})/(K_n\log K_n)\to\infty$ as $n\to\infty$.

\begin{theorem}\label{thm:conch_consistency_oracle}
Under the above conditions, $\max_{k\neq k_\star} p_{(k),n} \xrightarrow{P} 0$.
Consequently, we have that
\[
\P\bigl(\mathcal{K}^{\croc}_{n,1-\alpha} = \{k_\star\}\bigr) \ge 1-\alpha-\mathrm{o}(1).
\]
\end{theorem}

The proof is deferred to Appendix~\ref{app:uniform_decay_croc_pvalue}. This result shows that, under mild and interpretable conditions, the oracle LLR score yields asymptotically sharp localization of the root-cause index.

\section{\croc{} under cross-stream dependence}\label{sec:croc_dependence}
\IncMargin{1.2em}
\begin{algorithm}[t]
\caption{$\croc$-dep: CROC under structured cross-stream dependence}
\label{alg:croc-dep}
\KwIn{$\bX \in \prod_{k=1}^K \Xcal_k^n$, $\Rcal$, partition $\{G_1,\dots,G_M\}$, $1-\alpha$}
\KwOut{$\mathcal{K}^{\croc\text{-}\mathrm{dep}}_{1-\alpha}$}

\For{$\bt \in \Rcal$}{
    \For{$m \in [M]$}{
        Compute group-level $p$-value $p^{[G_m]}_{\bt}$ via $\croc$ restricted to $\bX^{[G_m]}$\;
    }
    $p_{\bt} \gets \min\{1,\, M \min_{m\in[M]} p^{[G_m]}_{\bt}\}$\;
}
\For{$k \in [K]$}{
    $I_k \gets \{\bt \in \Rcal : t_k < t_j \text{ for all } j\neq k\}$\;
    $p_{(k)} \gets \max_{\bt \in I_k} p_{\bt}$\;
}
$\mathcal{K}^{\croc\text{-}\mathrm{dep}}_{1-\alpha} \gets \{k\in[K]: p_{(k)} > \alpha\}$\;
\Return{$\mathcal{K}^{\croc\text{-}\mathrm{dep}}_{1-\alpha}$}
\end{algorithm}
\DecMargin{1.2em}
\Cref{assn:exchangeability_multi} assumes independence across streams, which may not hold in many applications where streams interact and changes in one can influence others, inducing dependence across streams. We therefore extend \croc{} to accommodate arbitrary cross-stream dependence, while retaining minimal within-stream structure: for each stream, the pre- and post-change segments are independent and separately exchangeable.

Let $\bX^{[k]} := (X_{1,k},\ldots,X_{n,k})$. Using the notation from Section~\ref{sec:croc_defn}, the null $\Hcal'_{0,\bt}$ decomposes as $\Hcal'_{0,\bt} = \bigcap_{k=1}^K \Hcal^{[k]}_{0,t_k}$, where $\Hcal^{[k]}_{0,r}$ asserts that $(X_{1,k},\ldots,X_{r,k})$ and $(X_{r+1,k},\ldots,X_{n,k})$ are independent and separately exchangeable. Under cross-stream dependence, the $p$-value in~\eqref{eq:pval_croc} is no longer directly applicable. However, each $\Hcal^{[k]}_{0,t_k}$ can be tested via the single-stream \conch{} procedure, yielding $p$-values $p^{[k]}_{t_k}$ from~\eqref{eq:pvalue_conch}. In the absence of knowledge of the dependence structure, we combine these via Bonferroni: $p_\bt = \min\bigl\{1,\, K \min_{k\in[K]} p^{[k]}_{t_k}\bigr\}, \qquad \bt \in \Rcal$.
Aggregating as in~\eqref{eq:pval_croc_agg} gives $p_{(k)}$, which remains valid for $\Hcal_{0,k}$ since each $p^{[k]}_{t_k}$ is valid (Section~\ref{sec:conch_single}). This yields the confidence set $\mathcal{K}^{\conch\text{-}\mathrm{agg}}_{1-\alpha} := \{k\in[K]: p_{(k)} > \alpha\}$,
which we refer to as the \conch{}-agg procedure, summarized in \Cref{alg:conch-agg}. Its finite-sample validity follows directly from that of \conch{}.

The $\conch$-agg procedure serves as a baseline method for root-cause analysis under arbitrary cross-stream dependence. While it guarantees valid inference without any knowledge of the dependence structure, the Bonferroni correction can be overly conservative, leading to wide confidence sets.

In many applications, however, partial information about the dependence structure is available. In particular, interactions across streams are often localized, allowing the streams to be partitioned into groups that are internally independent. This motivates a refinement of $\conch$-agg: by applying the full $\croc$ procedure within each independent group and aggregating across groups, one can leverage such structure to obtain substantially sharper confidence sets while retaining distribution-free validity. We develop this dependence-aware extension, termed $\croc$-dep, presented in \Cref{alg:croc-dep} and described in Appendix~\ref{app:dependent_stream}, along with numerical comparisons demonstrating its improved efficiency.

\section{Experiments}\label{sec:experiments}
In this section, we evaluate \croc{} via synthetic experiments. All \croc{} $p$-values are computed using the Monte Carlo approximation~\eqref{eq:pvalue_croc_MC} with $M=100$.

\subsection{Root-cause analysis in Gaussian mean-shift setting}\label{sec:gaussian_root_cause}
\begin{figure}[t]
\centering
\begin{minipage}{0.49\textwidth}
    \centering
    \includegraphics[width=\linewidth]{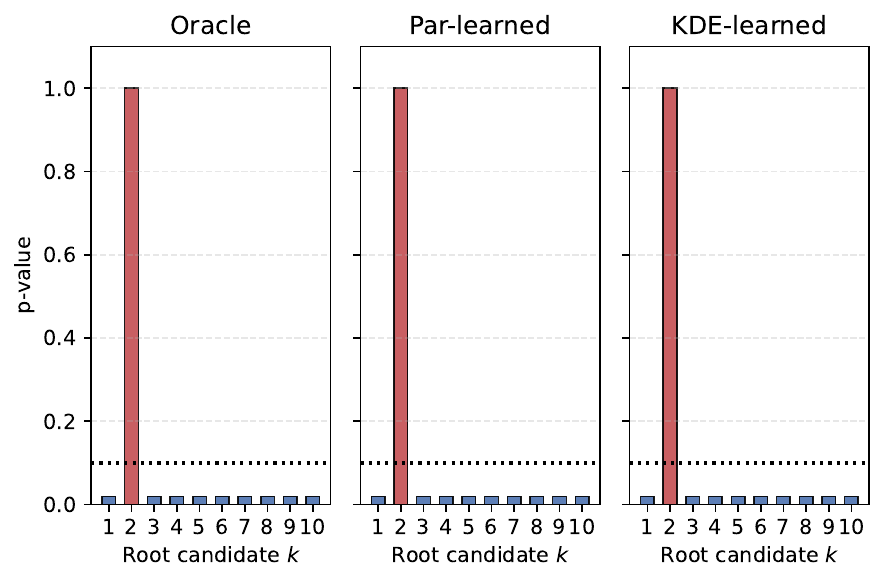}
\end{minipage}
\hfill
\begin{minipage}{0.49\textwidth}
    \centering
    \includegraphics[width=\linewidth]{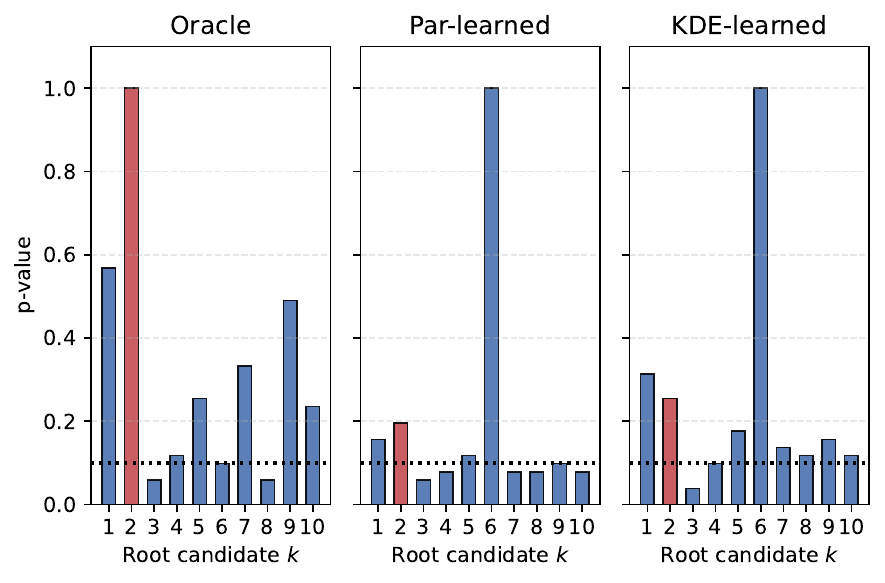}
\end{minipage}
\caption{CROC root-cause $p$-values across candidate streams. Left: Setting 1 (moderate separation). Right: Setting 2 (weak separation). The true root stream is highlighted in red, and the dotted line denotes $\alpha=0.1$.}
\label{fig:croc_root_gaussian}
\end{figure}

We evaluate $\croc$ on a multi-stream Gaussian mean-shift model to assess its efficiency in root-cause identification under varying signal strengths. We generate $K=10$ independent streams of length $n=80$. The root stream is $k^\star=2$ with changepoint $\xi_2=20$, while the remaining streams share a later changepoint at $\xi_k=50$ for all $k\neq 2$. The constraint set is then given by
\[
\Rcal=\{(i_1,\ldots,i_{10}) : \exists\, k\in[10]\ \text{such that}\ i_k=t_1<t_2=i_j,\ \forall j\neq k\}.
\]
For each stream $k$, $(X_{1,k},\ldots,X_{\xi_k,k})\overset{iid}{\sim} \mathcal{N}(\mu_{0,k},1)$ and $(X_{\xi_k+1,k},\ldots,X_{n,k})\overset{iid}{\sim} \mathcal{N}(\mu_{1,k},1)$,
where $\mu_{0,k}$ are drawn from an equi-spaced grid in $[-2,3]$, with $\mu_{1,2}=\mu_{0,2}+\delta_0$ for the root stream and $\mu_{1,k}=\mu_{0,k}+\delta_1$ otherwise. The parameters $(\delta_0,\delta_1)$ control the difficulty of root-cause 
identification. 
When $\delta_0$ is small relative to $\delta_1$, the signal in the root stream is weaker than in the subsequent streams, making it challenging to distinguish the true root cause from later changes. This challenge is further amplified since the latter share a common changepoint, so an algorithm may struggle to differentiate the root stream from the remaining streams. Motivated by this, we consider two regimes: \textbf{Setting 1} (moderate signal) with $\delta_0=1,\ \delta_1=2$, and \textbf{Setting 2} (weak signal) with $\delta_0=0.25,\ \delta_1=0.75$.

We evaluate $\croc$ using three CPP scores: (i) oracle log-likelihood ratio (LLR), (ii) parametrically learned LLR assuming Gaussianity, and (iii) nonparametrically learned LLR via kernel density estimation. Figure~\ref{fig:croc_root_gaussian} shows the \croc{} $p$-values for both settings, and highlight two key behaviors. In Setting 1, $\croc$ sharply isolates the true root stream, with its $p$-value exceeding the threshold while others remain below. In Setting 2, some non-root streams cross the threshold, reflecting increased difficulty; however, the true root stream is consistently retained in the confidence set.
Overall, $\croc$ remains reliable even when the root signal is weaker than subsequent changes, adaptively widening the confidence set in harder regimes while maintaining valid inclusion of the true root index.

\subsection{Root-cause analysis for image corruption detection}\label{sec:mnist_root_cause}

We evaluate $\croc$ in an image-based root-cause analysis setting using MNIST \citep{deng2012mnist}. We construct $K=5$ streams, each consisting of independent images drawn from fixed digit pairs—$(1,7)$, $(2,5)$, $(3,8)$, $(4,9)$, and $(0,6)$—that undergo a changepoint (corruption event) after which the images become blurred. Corruption is introduced by adding Gaussian noise with scale $\sigma$, yielding a semi-synthetic setup that mimics distributed systems where degradation may first arise in one stream and then propagate. The root stream (digits $(1,7)$) changes at time $t_1=25$ with mild noise $\sigma=0.1$, making the corruption visually subtle, while the remaining streams share a later changepoint at $t_2=50$ with stronger noise $\sigma=0.15$. This creates a challenging regime where the root signal is weaker than subsequent changes, and naive methods may misidentify the root.

For the CPP score, we pretrain a simple CNN on clean and heavily corrupted MNIST images ($\sigma=0.25$), so the learned representation is not tuned to detect the subtle corruption in the root stream. Despite these challenges, as shown in Figure~\ref{fig:croc_mnist_root}, $\croc$ constructs a confidence set that contains the true root index. Notably, the fifth stream attains the largest $p$-value, illustrating that reliance on point estimates can be misleading. In contrast, the conformal approach calibrates uncertainty and returns a valid confidence set, demonstrating robustness in high-dimensional, weak-signal settings.

\begin{figure}[t]
\centering
\begin{minipage}{0.49\textwidth}
    \centering
    \includegraphics[width=\linewidth]{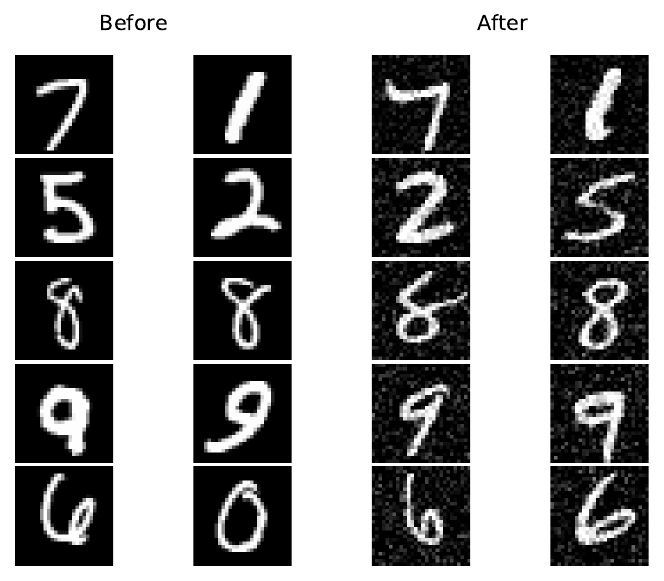}
\end{minipage}
\hfill
\begin{minipage}{0.49\textwidth}
    \centering
    \includegraphics[width=\linewidth]{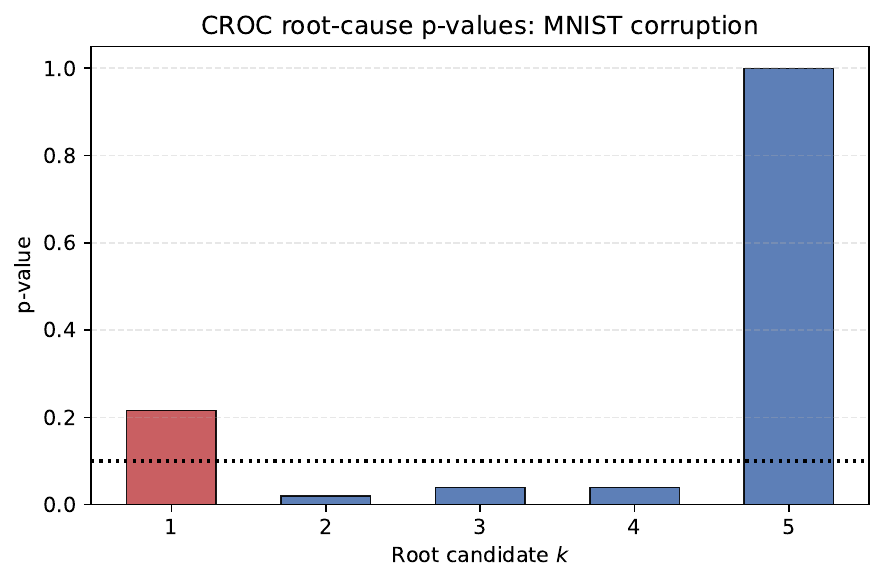}
\end{minipage}
\caption{Left: Example images from each stream before and after the changepoint. Right: $\croc$ root-cause $p$-values across streams. The true root stream is highlighted in red, and the dotted line denotes $\alpha=0.1$.}
\label{fig:croc_mnist_root}
\end{figure}

\section{Discussion}\label{sec:discussion}

We proposed \croc{}, a conformal framework that constructs finite-sample valid confidence sets for the root-cause index in a multi-stream setting under minimal exchangeability assumptions. Beyond validity, we established a universality property, showing that any distribution-free procedure can be represented within this framework via an appropriate choice of score function. We also characterized optimal score design under a canonical model and showed that practical approximations yield asymptotically sharp localization. Empirical results demonstrate that \croc{} performs effectively across a range of settings, including challenging regimes where the change in the root stream is difficult to detect. While we showed that our approach extends to dependent streams, improving efficiency under structured dependence and extending the framework to handle within-stream dependence remain important directions for future work.

\section*{Acknowledgments}
The authors acknowledge the funding from the Sloan Fellowship.

\bibliographystyle{chicago}
\bibliography{bibliography}
\appendix

\section{Additional experiment: root-cause analysis in multi-domain sentiment data}

We next demonstrate \croc{} on a multi-domain sentiment analysis task. We use the multi-domain sentiment dataset \citep{blitzer2007biographies}, which consists of reviews from four domains: books, DVD, electronics, and kitchen/housewares. Each domain is treated as an independent data stream with binary sentiment labels. 

To evaluate root-cause identification on this dataset, we construct a semi-synthetic setting in which the books domain undergoes a sentiment shift first, followed by a common later shift in the remaining three domains. This setup reflects realistic scenarios where changes in customer feedback for one product category propagate to related domains.

Concretely, we generate $K=4$ streams, each of length $n=150$. The books stream serves as the root stream with changepoint $\xi_1=50$, while the DVD, electronics, and kitchen/housewares streams share a later changepoint at $\xi_2=80$. Prior to the changepoint, reviews are $60\%$ positive; after the changepoint, they are $60\%$ negative. Representative examples include:
\begin{itemize}
    \item Positive: ``This book provides an excellent guide to finding and hiring the right people for your organization. \ldots''
    \item Positive: ``It does the hard work of picking out a crucial 16-bar piece of a song for you! Great for auditions. \ldots''
    \item Negative: ``The story was slow, predictable and boring. I am still amazed I hung in there \ldots''
    \item Negative: ``This truly is a horrible movie. Seven seconds to realize I made a mistake in renting it \ldots''
\end{itemize}

We use a DistilBERT model fine-tuned for sentiment classification \citep{sanh2019distilbert} (trained on the Stanford Sentiment Treebank dataset) to extract logits for each review, which are then used to construct the CPP score.

Figure~\ref{fig:croc_sentiment_root_cause} shows the resulting root-cause $p$-values. The books stream exhibits a clear peak, while the remaining streams lie below $\alpha=0.1$, indicating that \croc{} successfully identifies the domain where the sentiment shift first occurs.

\begin{figure}[!h]
    \centering
    \includegraphics[width=0.7\linewidth]{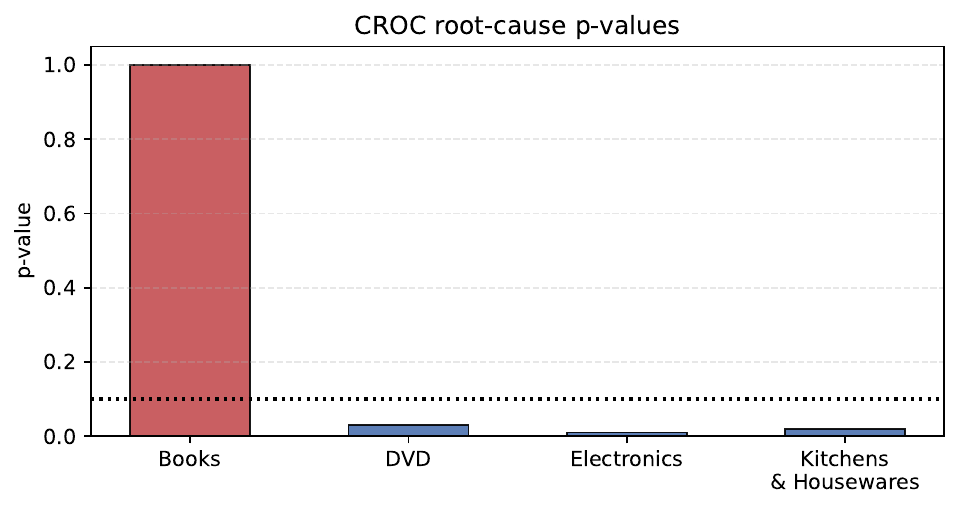}
    \caption{\croc{} root-cause $p$-values for the multi-domain sentiment experiment. The streams correspond to books, DVD, electronics, and kitchen/housewares reviews. The black dotted line denotes the threshold $\alpha=0.1$.}
    \label{fig:croc_sentiment_root_cause}
\end{figure}

\section{CROC under structured cross-stream dependence}\label{app:dependent_stream}

\IncMargin{1.2em}
\begin{algorithm}[t]
\caption{$\conch$-agg: Root cause analysis under arbitrary cross-stream dependence}
\label{alg:conch-agg}
\KwIn{$\bX \in \prod_{k=1}^K \Xcal_k^n$, $\Rcal$, $1-\alpha$}
\KwOut{$\mathcal{K}^{\conch\text{-}\mathrm{agg}}_{1-\alpha}$}

\For{$\bt \in \Rcal$}{
    \For{$k \in [K]$}{
        Compute single-stream $p$-value $p^{[k]}_{t_k}$ via \conch{} on $\bX^{[k]}$\;
    }
    $p_{\bt} \gets \min\{1,\, K \min_{k\in[K]} p^{[k]}_{t_k}\}$\;
}
\For{$k \in [K]$}{
    $I_k \gets \{\bt \in \Rcal : t_k < t_j \text{ for all } j\neq k\}$\;
    $p_{(k)} \gets \max_{\bt \in I_k} p_{\bt}$\;
}
$\mathcal{K}^{\conch\text{-}\mathrm{agg}}_{1-\alpha} \gets \{k\in[K]: p_{(k)} > \alpha\}$\;
\Return{$\mathcal{K}^{\conch\text{-}\mathrm{agg}}_{1-\alpha}$}
\end{algorithm}
\DecMargin{1.2em}

Building on the baseline method of \conch{}-agg framework (See \Cref{alg:conch-agg}) introduced in Section~\ref{sec:croc_dependence} for handling arbitrary cross-stream dependence, we show how additional knowledge of the dependence structure can be leveraged to improve over naive Bonferroni correction. 

Suppose the index set $[K]$ can be partitioned into disjoint groups $\mathcal{G}=\{G_1,\ldots,G_M\}$ such that $\bigcup_{m=1}^M G_m=[K]$ and $G_m\cap G_{m'}=\emptyset$ for $m\neq m'$, with the property that streams within each group are mutually independent, while arbitrary dependence may exist across groups.

For any candidate changepoint configuration $\bt=(t_1,\ldots,t_K)\in\Rcal$, recall that the null hypothesis can be written as $\Hcal'_{0,\bt}=\bigcap_{k=1}^K \Hcal^{[k]}_{0,t_k}$. Under the additional knowledge of these groups, this can be rewritten as $\Hcal'_{0,\bt}=\bigcap_{m=1}^M \Hcal^{[G_m]}_{0,\bt}$, where $\Hcal^{[G_m]}_{0,\bt}:=\bigcap_{k\in G_m}\Hcal^{[k]}_{0,t_k}$. For each group $G_m$, let $\bX^{[G_m]}:=(X_{i,k})_{i\in[n],\,k\in G_m}$ denote the corresponding data of streams in $G_m$. Since streams within $G_m$ are independent, $\Hcal^{[G_m]}_{0,\bt}$ can be tested directly using the \croc{} procedure restricted to the group $G_m$.

Accordingly, for each $\bt\in\Rcal$, we compute group-level conformal $p$-values
\[
p^{[G_m]}_{\bt}
:= \frac{1}{|\Pi^{[G_m]}_{\bt}|} \sum_{\pi\in \Pi^{[G_m]}_{\bt}}
\One{ S^{[G_m]}(\pi(\bX^{[G_m]}),\bt) \le S^{[G_m]}(\bX^{[G_m]},\bt)},
\]
where $\Pi^{[G_m]}_{\bt}$ denotes the split-permutation group acting jointly on the streams in $G_m$, and $S^{[G_m]}$ is a CPP score defined on $\bX^{[G_m]}$. Notably, the optimal score in~\eqref{eq:optimal_CPP_score} and its practical counterpart~\eqref{eq:practical_CPP_score} decompose additively across streams, which naturally induces group-level scores by restricting the summation to indices in $G_m$. Each $p^{[G_m]}_{\bt}$ is therefore a valid $p$-value for $\Hcal^{[G_m]}_{0,\bt}$.

We then combine these group-level $p$-values via a Bonferroni rule
\[
p_{\bt} := \min\bigl\{1,\; M \min_{m\in[M]} p^{[G_m]}_{\bt}\bigr\}, \qquad \bt\in\Rcal,
\]
which yields a valid $p$-value for the global null $\Hcal'_{0,\bt}$. Root-cause inference proceeds as in the main construction by defining $p_{(k)} := \max_{\bt\in I_k} p_{\bt}$ and lastly defining $\mathcal{K}^{\croc\text{-}\mathrm{dep}} := \{k\in[K]: p_{(k)}>\alpha\}$. The algorithm is summarized in \Cref{alg:croc-dep}.

This construction interpolates between two extremes: when all streams belong to a single group, it recovers the full \croc{} procedure; when each group is a singleton, it reduces to the \conch{}-agg construction. By exploiting partial independence, this approach can yield substantially sharper confidence sets compared to naive Bonferroni correction.

We illustrate this gain in a Gaussian mean-shift setting with $K=6$ streams of length $n=50$, using the oracle CPP score. The root stream ($k^\star=2$) undergoes an earlier changepoint at $t_1=15$, while the remaining streams share a later changepoint at $t_2=30$. Pre-change means are chosen from an equi-spaced grid in $[-3,3]$, and post-change means are shifted by $1$ for the root stream and $1.5$ for the others, with unit variance throughout.

We introduce structured dependence by simualating correlated data streams within steeam pairs $(1,2)$, $(3,4)$, and $(5,6)$ with correlation $\rho=0.65$, while the grouping $\{\{1,3,5\},\{2,4,6\}\}$ specifies sets of mutually independent streams. We compare the baseline \conch{}-agg procedure with \croc{}-dep, the dependence-aware variant. As shown in \Cref{fig:croc_structured_dep}, both methods retain coverage and include the true root-cause index. However, \conch{}-agg produces conservative confidence sets and fails to confidently exclude non-root streams. In contrast, \croc{}-dep leverages the known dependence structure to yield significantly sharper confidence sets, enabling more decisive identification of non-root streams.
This highlights the practicality of our \croc{} framework over the baseline aggregation of existing \conch{} confidence sets.

\begin{figure}[t]
    \centering
    \includegraphics[width=0.8\linewidth]{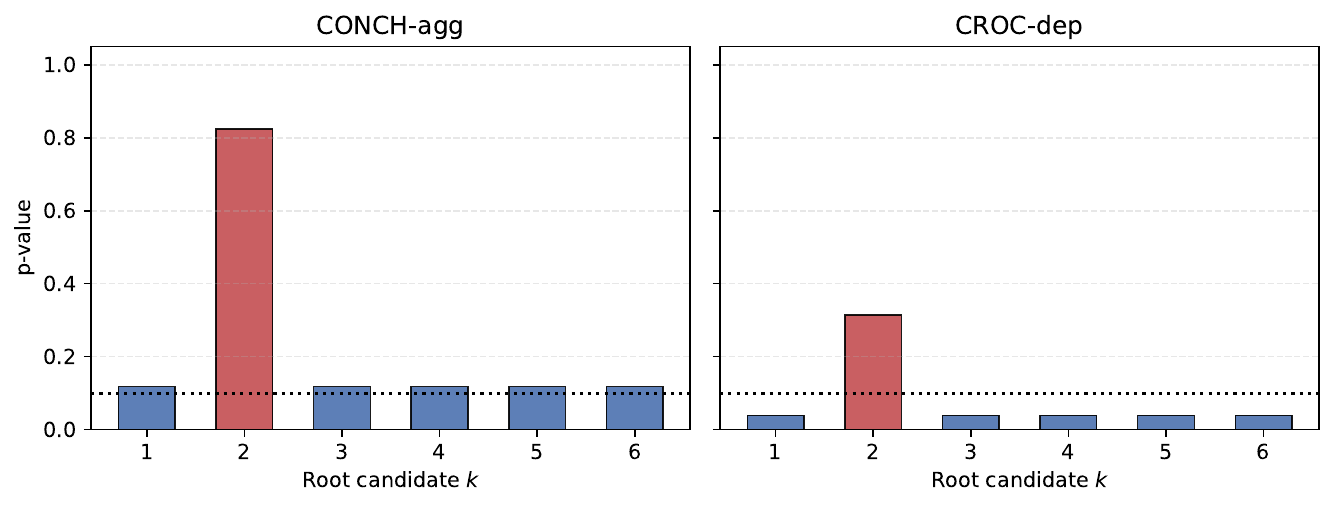}
    \caption{Comparison of \conch{}-agg and \croc{}-dep under structured cross-stream dependence with $K=6$ streams. The second stream is the root stream, and dependence is introduced within pairs $(1,2)$, $(3,4)$, and $(5,6)$. The plot shows root-cause $p$-values across candidate streams.}
    \label{fig:croc_structured_dep}
\end{figure}
\section{Proofs}
\subsection{Proving validity of \croc{} confidence sets}

\subsection{Proof of Theorem~\ref{thm:validity_croc}}
Fix $\bt\in \Rcal$.
Observe that under the null $\Hcal'_{0,\bt}$, $\pi(\bX)\overset{d}{=}\bX$ for any $\pi\in \Pi_{\bt}$. We define a function $p_{\bt}:\prod_{k=1}^K\Xcal_k^n\to [0,1]$ by
\[
p_{\bt}(\bx):=\frac{1}{|\Pi_{\bt}|}\sum_{\pi\in \Pi_{\bt}}\One{S(\pi(\bx),\bt)\leq S(\bx,\bt)},
\]
and note that $p_{\bt}\equiv p_{\bt}(\bX)$. Therefore, it follows that
\begin{align*}
    \P_{\Hcal'_{0,\bt}}\left(p_{\bt}(\bX)\leq \alpha\right)&=\frac{1}{|\Pi_{\bt}|}\sum_{\pi\in \Pi_{\bt}} \P_{\Hcal'_{0,\bt}}\left(p_{\bt}(\pi(\bX))\leq \alpha\right)\\
    &=\E_{\Hcal'_{0,\bt}}\left[\frac{1}{|\Pi_{\bt}|}\sum_{\pi\in \Pi_{\bt}}\One{p_{\bt}(\pi(\bX))\leq \alpha}\right]\\
    &=\E_{\Hcal'_{0,\bt}}\left[\frac{1}{|\Pi_{\bt}|}\sum_{\pi\in \Pi_{\bt}}\One{\frac{1}{|\Pi_{\bt}|}\sum_{\pi^\prime\in \Pi_{\bt}} \One{S(\pi^\prime(\bX), {\bt})\leq S(\pi(\bX),{\bt})}\leq \alpha}\right]\leq \alpha,
\end{align*}
where the penultimate step follows by noting that $\pi\circ \Pi_{\bt}=\Pi_{\bt}$, and the last step is a deterministic inequality. 

Consequently, since $\Hcal_{0,k}=\bigcup_{\bt\in I_k}\Hcal'_{0,\bt}$, recalling that $p_{(k)}=\max_{\bt\in I_k} p_{\bt}$, it follows that $\P_k(p_{(k)}\le \alpha)\le \alpha$.
This completes the proof.  $\hfill\mathsf{\square}$

\subsubsection{Proof of Theorem~\ref{thm:validity_croc_randz}}

Given permutations $\pi_{1,\bt},\ldots,\pi_{M,\bt}\in \Pi_\bt$, we define the function
\[
\tilde{p}_\bt(\bx;\pi_{1,\bt},\ldots,\pi_{M,\bt}):=\frac{1+\sum_{k=1}^M\One{S(\pi_{k,t}(\bx),\bt)\leq S(\bx,\bt)}}{1+M},
\]
Consider an additional uniform draw $\pi_{0,\bt}$ from $\Pi_\bt$. Note that with $\pi_{1,\bt},\ldots,\pi_{M,\bt}\overset{iid}{\sim}\textnormal{Unif}(\Pi_\bt)$, we have that marginally,
\[
(\pi_{1,\bt},\ldots,\pi_{M,\bt})\overset{d}{=}(\pi_{0,\bt}\circ\pi_{1,\bt},\ldots,\pi_{0,\bt}\circ\pi_{M,\bt}).
\]
Moreover, conditional on $\pi_{0,\bt},\pi_{1,\bt},\ldots,\pi_{M,\bt}$, $\bX\overset{d}{=}\pi_{0,\bt}(\bX)$ under the null $\mathcal{H}_{0,\bt}$. Consequently,
\begin{multline*}
    \tilde{p}_\bt(\bX;\pi_{1,\bt},\ldots,\pi_{M,\bt})\overset{d}{=}\tilde{p}_\bt(\bX;\pi_{0,\bt}\circ\pi_{1,\bt},\ldots,\pi_{0,\bt}\circ\pi_{M,\bt})\overset{d}{=}\tilde{p}_\bt(\pi_{0,\bt}(\bX);\pi_{0,\bt}\circ\pi_{1,\bt},\ldots,\pi_{0,\bt}\circ\pi_{M,\bt}).
\end{multline*}
Finally, note that for $\tilde{p}_\bt$, defined in \eqref{eq:pvalue_croc_MC},  $\tilde{p}_\bt\equiv \tilde{p}_\bt(\bX;\pi_{1,\bt},\ldots,\pi_{M,\bt})$, and therefore,
\begin{align*}
        \tilde{p}_\bt(\bX;\pi_{1,\bt},\ldots,\pi_{M,\bt})&\overset{d}{=}\tilde{p}_\bt(\pi_{0,\bt}(\bX);\pi_{0,\bt}\circ\pi_{1,\bt},\ldots,\pi_{0,\bt}\circ\pi_{M,\bt})\\
        &=\frac{1+\sum_{k=1}^M\One{S(\pi_{k,\bt}(\bX),\bt)\leq S(\pi_{0,\bt}(\bX),\bt)}}{M+1}\\
        &=\frac{\sum_{k=0}^M\One{S(\pi_{k,\bt}(\bX),\bt)\leq S(\pi_{0,\bt}(\bX),\bt)}}{M+1},
    \end{align*}
i.e., the rank of $S(\pi_{0,\bt}(\bX),\bt)$ in the exchangeable collection $\{S(\pi_{0,\bt}(\bX),\bt),S(\pi_{1,\bt}(\bX),\bt),\ldots, S(\pi_{M,\bt}(\bX),\bt)\}$. Consequently,
\[
\P_{\Hcal'_{0,\bt}}\left(\tilde{p}_\bt=\tilde{p}_\bt(\bX;\pi_{1,\bt},\ldots,\pi_{M,\bt})\leq \alpha\right)\leq \alpha.
\]
This proves the first part. Next, since $\Hcal_{0,k}=\bigcup_{\bt\in I_k}\Hcal'_{0,\bt}$, recalling that $\tilde{p}_{(k)}=\max_{\bt\in I_k} \tilde{p}_{\bt}$, it follows that $\P_k(\tilde{p}_{(k)}\le \alpha)\le \alpha$.
This completes the proof.  $\hfill\mathsf{\square}$

\subsubsection{Achieving exact validity for testing $\Hcal'_{0,\bt}$}\label{app:exact_coverage}
While both $p$-values $p_t$ and $\tilde{p}_t$ in \eqref{eq:pval_croc} and \eqref{eq:pvalue_croc_MC} control the Type~I error under the null $\Hcal'_{0,t}$ at level~$\alpha$, it is sometimes desirable to attain \emph{exact} level-$\alpha$ validity. Even though we are doing a max-aggregation afterwards, as in~\eqref{eq:pval_croc_agg}, achieving exact validity for testing $\Hcal'_{0,\bt}$ can yield more powerful or sharper procedures. To this end, we introduce a randomized refinement of the $p$-values that guarantees exact validity.
Specifically, define
\begin{equation}\label{eq:pvalue_conch_randomized}
    \bar{p}_\bt := \frac{1}{|\Pi_\bt|}\sum_{\pi\in \Pi_\bt} \One{S(\pi(\bX,\bt)) < S(\bX,\bt)}
    + U \cdot \frac{1}{|\Pi_\bt|}\sum_{\pi\in \Pi_\bt} \One{S(\pi(\bX),\bt) = S(\bX,\bt)},
\end{equation}
where $U \sim \textnormal{Unif}[0,1]$. 
\begin{theorem}\label{thm:coverage_exactly_1-alpha}
    For each $t \in \Rcal$, $\bar{p}_t$ defined in \eqref{eq:pvalue_conch_randomized} is a valid $p$-value under $\Hcal'_{0,\bt}$. In particular, for any $\alpha \in (0,1)$,
    \[
        \P_{\Hcal'_{0,\bt}}\!\left(\bar{p}_{\bt} \le \alpha\right) = \alpha.
    \]
\end{theorem}
\begin{proof}
    We begin by letting $F$ denote the distribution of $S(\pi(\bX),\bt)$ with $\pi \sim \textnormal{Unif}(\Pi_\bt)$, conditional on the multisets $M_{\textnormal{left}}^{(1)},\ldots, M_{\textnormal{left}}^{(K)}$  and $M_{\textnormal{right}}^{(1)},\ldots, M_{\textnormal{right}}^{(K)}$ where we write for each $k\in [K]$,
    \[
    M_{\textnormal{left}}^{(k)}:=\{X_{1,k},\ldots,X_{t_k,k}\},\qquad M_{\textnormal{right}}^{(k)}:=\{X_{t_k+1,k},\ldots,X_{n,k}\}.
    \]
Then, we observe
\[
\bar{p}_\bt = \lim_{y \uparrow S(\bX,\bt)} F(y) + U \bigl(F(S(\bX,\bt)) - \lim_{y \uparrow S(\bX,\bt)} F(y)\bigr).
\]
Under $\Hcal'_{0,\bt}$, conditional on the multisets $\bigl\{M_{\textnormal{left}}^{(k)}\bigr\}_{k=1}^K$  and $\bigl\{M_{\textnormal{right}}^{(k)}\bigr\}_{k=1}^K$, we have $S(\bX,\bt) \overset{d}{=} S(\pi(\bX),\bt)$.
Hence, by \citet[Lemma~E.1]{dandapanthula2025offline}, the $p$-value $\bar{p}_t$, conditional on the same multisets, follows $\textnormal{Unif}[0,1]$ (see also \citealp{brockwell2007universal}). 
Therefore,
\[
\P_{\Hcal'_{0,\bt}}(\bar{p}_\bt \le \alpha) 
= \E_{\Hcal'_{0,\bt}}\!\left[\P_{\Hcal'_{0,\bt}}\!\left(\bar{p}_\bt \le \alpha \mid \bigl\{M_{\textnormal{left}}^{(k)}\bigr\}_{k=1}^K, \bigl\{M_{\textnormal{right}}^{(k)}\bigr\}_{k=1}^K\right)\right]
= \E_{\Hcal'_{0,\bt}}[\alpha] = \alpha.
\]
This completes the proof.
\end{proof}

\subsection{Proof of universality (Theorem~\ref{thm:universality_croc})}

Fix $n\in \N$. First, given the confidence set $C$, we consider the CPP score
\[
S(\bx,\bt) = \One{\argmin_{k\in [K]}t_k \in C(\bx)} \in \{0,1\},
\]
We will show that the \croc{} confidence set constructed from this score, denoted by $\mathcal{K}^{\croc}_{1-\alpha}$, coincides exactly with the given confidence set~$C$.

To that end, we start with showing that $\mathcal{K}^{\croc}_{1-\alpha}(\bX) \supseteq C(\bX)$; that is, if for some $k\in [K]$, $k \in C(\bX)$, then $p_{(k)}>\alpha$. In fact, it suffices to show that for some $\bt\in I_k$, $p_{\bt}>\alpha$. This in fact holds by choosing any $\bt\in I_k$.
Then, with the above choice of CPP score, $S(\bX,\bt) = 1$, and consequently,
\[
p_{\bt} = \frac{1}{|\Pi_\bt^K|}\sum_{\pi \in \Pi_\bt^K} \One{S(\pi(\bX),\bt) \leq S(\bX,\bt)}
     = \frac{1}{|\Pi_\bt^K|}\sum_{\pi \in \Pi_\bt^K} \One{S(\pi(\bX),\bt) \leq 1} = 1.
\]

Next, we need to show that $\mathcal{K}^{\croc}_{1-\alpha}(\bX) \subseteq C(\bX)$, i.e., if $k \notin C(\bX)$, then $p_{(k)} \leq \alpha$. To that end, we first claim that for any $k \in [K]$, for any $\bt \in I_k$, and any $\bx \in \prod_{k=1}^K\mathcal{X}^{n}$, it holds that
\begin{equation}\label{eq:claim}
    \frac{1}{|\Pi_\bt|}\sum_{\pi \in \Pi_\bt} \One{k \in C(\pi(\bx))} \geq 1 - \alpha.
\end{equation}
Let us first show that how this claim implies the main result. Observe that if $k \notin C(\bX)$, then for any $\bt \in I_k$, we have $S(\bX,\bt) = 0$. Consequently, for any such $\bt\in \Rcal$
\begin{align*}
p_t &= \frac{1}{|\Pi_t|}\sum_{\pi \in \Pi_t} \One{S(\pi(\bX),\bt) \leq S(\bX,\bt)} \\
    &= \frac{1}{|\Pi_t|}\sum_{\pi \in \Pi_t} \One{S(\pi(\bX),\bt) \leq 0}
     = \frac{1}{|\Pi_t|}\sum_{\pi \in \Pi_t} \One{k\notin C(\pi(\bX))} \leq \alpha,
\end{align*} 
where the last step follows from \eqref{eq:claim}. This proves the other direction.

Now, it is left to prove the claim~\eqref{eq:claim}. Fix any $\bt \in I_k$ and sample $\pi=(\pi_1,\ldots,\pi_K)$ uniformly from the set of permutations $\Pi_\bt$. Define $\tilde{\bX}:= (\tilde{X}_{i,j})_{i\in [n],j\in [K]}$ such that for each $j\in [K]$, 
\[
(\tilde{X}_{1,j},\ldots,\tilde{X}_{n,j})=\pi_j((x_{1,j},\ldots,x_{n,j})).
\]
 Conditional on the multisets $\bigl\{\{x_{1,j}, \ldots, x_{t_j,j}\},\{x_{t_j+1,j}, \ldots, x_{n,j}\}: j=1,\ldots, K\bigr\}$, we have for each $j\in [K]$ that
\[
\tilde{X}_{1,j}, \ldots, \tilde{X}_{t,j} \text{ are exchangeable, and } \tilde{X}_{t+1,j}, \ldots, \tilde{X}_{n,j} \text{ are exchangeable}, 
\]
and that $(\tilde{X}_{1,j}, \ldots, \tilde{X}_{t,j})$ and $(\tilde{X}_{t+1,j}, \ldots, \tilde{X}_{n,j})$ are independent, implying that the sampling process of $\tilde{\bX}$ satisfies the modeling assumptions imposed by the null $\Hcal'_{0,t}$. Consequently, for the data stream $(\tilde{\bX})$, $k_\star=k$, and by the theorem statement
\[
\P_{\pi \sim \text{Unif}(\Pi_t)}\big(k \in C(\tilde{\bX}) \,\big|\, \bigl\{\{x_{1,j}, \ldots, x_{t_j,j}\},\{x_{t_j+1,j}, \ldots, x_{n,j}\}: j=1,\ldots, K\bigr\}\big) \geq 1 - \alpha,
\]
or equivalently, \eqref{eq:claim} holds. This completes the proof. $\hfill\mathsf{\square}$

\subsection{Proving properties of the CPP score and deriving its optimal form}

\subsubsection{Proof of Proposition~\ref{prop:score-properties}}
Fix $t\in \Rcal$. If $S$ satisfies $S(\cdot,\bt)=S(\pi(\cdot),\bt)$ for all $\pi\in \Pi_t$, then by \eqref{eq:pval_croc} note that $p_\bt$ is identically equal to $1$. This proves the first part.

For the second part, fix $t\in \Rcal$. By \eqref{eq:pval_croc},
\[
p_{\bt,1}=\frac{1}{|\Pi_\bt|}\sum_{\pi\in \Pi_\bt} \One{S(\pi(\bX),\bt)\leq S(\bX,\bt)},\quad 
p_{\bt,2}=\frac{1}{|\Pi_\bt|}\sum_{\pi\in \Pi_\bt} \One{f(S(\pi(\bX),\bt))\leq f(S(\bX,\bt))}.
\]
Since $f$ is non-decreasing, $S(\pi(\bX),\bt)\leq S(\bX,\bt)$ implies $f(S(\pi(\bX),\bt))\leq f(S(\bX,\bt))$,
and therefore $p_{\bt,1}\leq p_{\bt,2}$. As this holds for all $\bt\in \Rcal$, it further follows that $C_1\subseteq C_2$.  This proves the second part.
$\hfill\mathsf{\square}$

\subsubsection{Proving optimal form of CPP score}

We recall that the \croc{} algorithm (\Cref{alg:croc}) first tests the null $\Hcal'_{0,\bt}$ for each $\bt\in \Rcal$ via the $p$-value~\eqref{eq:pval_croc}, and then aggregates these values as in~\eqref{eq:pval_croc_agg}. Consequently, the problem of defining an optimal CPP score for root-cause analysis reduces to constructing powerful tests for the null hypotheses $\Hcal'_{0,\bt}$. In particular, if $\bxi$ denotes the true changepoint configuration, we aim to design tests that reject $\Hcal'_{0,\bt}$ for $\bt\neq \bxi$ with high power. 

Since \croc{} relies on conformal $p$-values, we analyze the testing problem of $\Hcal'_{0,\bt}$ against $\Hcal'_{0,\bxi}$ conditional on the relevant multisets. Specifically, fix any $\bt\in\Rcal$. For each $j\in[K]$, define
\[
\mathcal{X}_{L,\bt,j}:=\{X_{1,j},\ldots,X_{t_j,j}\}, \quad 
\mathcal{X}_{R,\bt,j}:=\{X_{t_j+1,j},\ldots,X_{n,j}\},
\]
as the (unordered) left and right multisets for the $j$th stream. For any $\mathbf{r}\in \Rcal$, let $\Pcal^{(\mathbf{r})}$ denote the conditional distribution of $\bX$ given $(\mathcal{X}_{L,\bt,1},\mathcal{X}_{R,\bt,1},\ldots,\mathcal{X}_{L,\bt,K},\mathcal{X}_{R,\bt,K})$ under $\Hcal'_{0,\mathbf{r}}$, and write $\Hcal'_{\mathbf{r}}$ to hypothesize that
\[
\bX\mid(\mathcal{X}_{L,\bt,1},\mathcal{X}_{R,\bt,1},\ldots, \mathcal{X}_{L,\bt,K},\mathcal{X}_{R,\bt,K})\sim\Pcal^{(\mathbf{r})}.
\]

Restricting attention to conformal $p$-values, we follow the standard construction of \croc{}. Given a CPP score $S:(\prod_{k=1}^K \Xcal_k^n)\times \Rcal \to \R$, define $s(\cdot)=S(\cdot,\bt)$ and let $\Pi_\bt$ denote the permutation set, the randomized conformal $p$-value (cf.~\eqref{eq:pvalue_conch_randomized}) is
\begin{equation}\label{eq:conformal_pvalue}
p_\bt(s) =\frac{1}{|\Pi_\bt|}\sum_{\pi\in\Pi_\bt}\One{s(\pi(\bX))<s(\bX)}\ +\ U\cdot \frac{1}{|\Pi_\bt|}\sum_{\pi\in\Pi_\bt}\One{s(\pi(\bX))=s(\bX)},
\end{equation}
where $U\sim\mathrm{Unif}(0,1)$ is independent of $\bX$ and the permutations.

The analysis in Section~\ref{app:exact_coverage} establishes that the test $\phi_\bt(\bX;s)=\One{p_\bt(s)\leq \alpha}$ controls Type~I error exactly at level $\alpha$ under $\Hcal'_{0,\bt}$ for any score function $s$. We therefore seek an optimal score $s^\star$ such that the corresponding test maximizes power against $\Hcal'_{0,\bxi}$. The following lemma, a conformal analogue of the Neyman--Pearson lemma, is the key ingredient of our analysis.

\begin{lemma}\label{lem:conformal_np_lemma}
Fix $\bt_1=(t_{11},\ldots,t_{1K}),\bt_2=(t_{21},\ldots,t_{2K})\in \Rcal$ with $\bt_1\neq \bt_2$. The power $\E_{\Hcal^\prime_{\bt_2}}[\phi_{\bt_1 }(\bX;s)]$ is maximized by the score function
\[
s^\star(\bx) := \prod_{j=1}^K\left(\frac{\prod_{i\leq t_{1j}}f_{0,j}(x_{i,j})\cdot \prod_{i>t_{1j}} f_{1,j}(x_{i,j})}{\prod_{i\leq t_{2j}}f_{0,j}(x_{i,j})\cdot \prod_{i>t_{2j}} f_{1,j}(x_{i,j})}\right).
\]
\end{lemma}
\begin{proof}
In the above setup, we consider the following hypothesis testing problem:
\begin{align*}
&\Hcal^\prime_{0}:\bX\mid(\mathcal{X}_{L,\bt_1,1},\mathcal{X}_{R,\bt_1,1},\ldots, \mathcal{X}_{L,\bt_1,K},\mathcal{X}_{R,\bt_1,K})\sim\Pcal^{(\bt_1)}
\quad \text{vs.} \\
&\qquad\Hcal^\prime_{1}:\bX\mid(\mathcal{X}_{L,\bt_1,1},\mathcal{X}_{R,\bt_1,1},\ldots, \mathcal{X}_{L,\bt_1,K},\mathcal{X}_{R,\bt_1,K})\sim\Pcal^{(\bt_2)}.
\end{align*}

Given samples $\bX\in(\prod_{k=1}^K \Xcal_k^n)$, observe that
\[
\frac{\mathsf{d}\bigl(\mathcal{P}^{(\bt_2)}\bigr)}{\mathsf{d}\bigl(\mathcal{P}^{(\bt_1)}\bigr)}(\bX)\propto 
\prod_{j=1}^K\left(\frac{\prod_{i\leq t_{2j}}f_{0,j}(x_{i,j})\cdot \prod_{i>t_{2j}} f_{1,j}(x_{i,j})}{\prod_{i\leq t_{1j}}f_{0,j}(x_{i,j})\cdot \prod_{i>t_{1j}} f_{1,j}(x_{i,j})}\right)
= {s^\star(\bX)}^{-1}.
\]

By the Neyman--Pearson lemma \citep[Theorem~3.2.1~(ii)]{lehmann2005testing}, any test $\phi(\bX)$ that attains exact level $\alpha$ under $\Hcal_{0}^\prime$ and satisfies
\begin{equation}\label{eqn:NP_optimal_form}
    \phi(\bX)=
\begin{cases}
    1 & \text{if } {s^\star(\bX)}^{-1} > \tau_{\alpha},\\[4pt]
    0 & \text{if } {s^\star(\bX)}^{-1} < \tau_{\alpha},
\end{cases}
\end{equation}
for some threshold $\tau_{\alpha}\in \R$, is most powerful for testing $\Hcal_{0}^\prime$ against $\Hcal_{1}^\prime$.

From Section~\ref{app:exact_coverage}, we know that the test $\phi_\bt(\cdot;s)=\One{p_\bt(s)\leq \alpha}$ controls Type~I error exactly at level $\alpha$ under $\Hcal_{0}^\prime$ for any score function $s$. Therefore, to establish optimality of $s^\star$, it suffices to show that $\phi_\bt(\cdot;s^\star)$ has the form given in \eqref{eqn:NP_optimal_form}.

Let $\bX_{\pi}=\pi(\bX)$ for $\pi\sim \mathrm{Unif}(\Pi_\bt)$, and let $F_{s^\star(\bX_\pi)}$ denote the conditional cumulative distribution function of $s^\star(\bX_{\pi})$ given $\bX$. Define
\[
\tau_{\alpha}:=\inf\{y\in \R: F_{s^\star(\bX_\pi)}(y)\geq \alpha\}.
\]
By the definition of $p_\bt$ in \eqref{eq:conformal_pvalue}, we have
\begin{align*}
    {s^\star(\bX)}^{-1} > \tau_{\alpha} &\implies p_\bt(s^\star)\leq \alpha,\\
    {s^\star(\bX)}^{-1} < \tau_{\alpha} &\implies p_\bt(s^\star)> \alpha,
\end{align*}
which establishes the desired form. This completes the proof.
\hfill$\square$
\end{proof}

\paragraph{Proof of Theorem~\ref{thm:optimal_score}}
We begin by noting that only the coordinate $S(\cdot,\bt)$ determines the conformal $p$-value $\bar{p}_\bt$ in \eqref{eq:pvalue_conch_randomized}. Applying Lemma~\ref{lem:conformal_np_lemma} with $\bt_1=\bt$ and $\bt_2=\bxi$ yields the optimal form of the score. Finally, since $\log$ is an increasing transformation, Proposition~\ref{prop:score-properties} gives the expression for $S^\textnormal{OPT}(\cdot,\bt)$.
\hfill$\square$

\subsection{Proving asymptotic sharpness of $\croc$ confidence set}\label{app:consistency}
In this section, we give the proof for the asymptotic sharpness of \croc{} confidence set computed with the oracle LLR score. First, we start with a few notation.

 Recall that $\ell_k(x)=f_{0,k}(x)/f_{1,k}(x)$, and hence observe that
$\mathrm{KL}(P_{0,k}\|P_{1,k})=\E_{X\sim P_{0,k}}[\ell_k(X)]$ and  $\mathrm{KL}(P_{1,k}\|P_{0,k})=\E_{X\sim P_{1,k}}[-\ell_k(X)],$
where $\mathrm{KL}(P\|Q)$ denotes the Kullback-Leibler divergence between distributions $P$ and $Q$. Moreover, we define the Jeffreys divergence as
\[
\mathrm{J}(P_{0,k},P_{1,k})=\mathrm{KL}(P_{0,k}\|P_{1,k})+\mathrm{KL}(P_{1,k}\|P_{0,k}).
\]
Accordingly, we write
\[
J_{\min}=\min_{k\in \N} \mathrm{J}(P_{0,k},P_{1,k}),\qquad J_{\max}=\max_{k\in \N} \mathrm{J}(P_{0,k},P_{1,k}).
\]
We assume throughout that $0<J_{\min}<J_{\max}<\infty$.
The corresponding var-entropy measures are given by $\sigma_{0,k}^2:=\mathrm{Var}_{X\sim P_{0,k}}(\ell_k(X))$ and $
\sigma_{1,k}^2:=\mathrm{Var}_{X\sim P_{1,k}}(\ell_k(X))$, and we define \[
\sigma_\star= \max_{k\in [\N]}\max\{\sigma_{0,k},\sigma_{1,k}\}.\]
We assume that $0<\sigma_\star<\infty$.
Next, given an estimate of $\ell_k$, for any $k\in [n-1]$, we define the empirical averages
where
we write for $\bt_1=(t_{11},\ldots,t_{1K})\in \Rcal$ and $\bt_2=(t_{21},\ldots,t_{2K})\in \Rcal$,
\begin{align*}
    &\Gamma(\bt_1,\bt_2):=\sum_{k=1}^K |t_{1k}-t_{2k}|\cdot\left(\frac{\one{t_{1k}>t_{2k}}}{t_{1k}}\sum_{i=1}^{t_{1k}} \ell_k(X_{i,k,n}) -\frac{\one{t_{1k}<t_{2k}}}{n-t_{1k}}\sum_{i=t_{1k}+1}^n \ell_k(X_{i,k,n}) \right),~~\text{and}\\
    &\hspace{4cm}\hat{v}_{n,k}:=\frac{1}{n}\sum_{i=1}^{n}\ell^{\,2}_k(X_{i,k,n}),\qquad k\in \N.
\end{align*}

Now, instead of working with the true \croc{} $p$-values directly, we would work with an approximation of the same, which is derived by freezing the MLE estimate. We first define them formally. Throughout, we write $\hat{\xi}_n\equiv \hat{\xi}_n(\bX)$ to denote the MLE on the original data defined in \eqref{eq:mle}. Then, the oracle CPP score admits the simplified expression
\[
S(\bx,\bt)=\sum_{k=1}^K {\rm sgn}(t_k-\hat{\xi}_{k,n}(\bx))\cdot\biggl\{\sum_{i=t_k\wedge \hat{\xi}_{k,n}(\bx)+1}^{t_k\vee \hat{\xi}_{k,n}(\bx)}\ell_k(x_{i,k,n})\biggr\}.
\]
This means that for computing $S(\pi(\bx),\bt)$, one needs to compute the MLE again on the permuted data, i.e., $\hat{\xi}_{k,n}(\pi(bx))$. This then gives us the conformal $p$-values $p_{\bt,n}$ for each $\bt\in \Rcal$, and once aggregated, we get the final $p$-values $p_{(k),n}$. To construct an intermediate approximation to these $p$-values, we freeze the MLE on the original data, and do not recompute it again for the permuted data streams. In particular, consider the score function,
\begin{equation}\label{eqn:score-for_consistency}
S(\bx,\bt)=\sum_{k=1}^K {\rm sgn}(t_k-\hat{\xi}_{k,n})\cdot\biggl\{\sum_{i=t_k\wedge \hat{\xi}_{k,n}+1}^{t_k\vee \hat{\xi}_{k,n}}\ell_k(x_{i,k})\biggr\}
\end{equation}
Based on this score, we define the permutation $p$-values $(\tilde{p}_{1,n},\ldots,\tilde{p}_{n-1,n})$,
\begin{equation}\label{eqn:pvalue_for_consistency}
\tilde{p}_{t,n}
:=\frac{1}{|\Pi_t|}\sum_{\pi\in \Pi_t}\,\One{S_t^{(n)}(\pi(\bX))\le S_t^{(n)}(\bX)}.
\end{equation}
The only distinction between these $p$-values and the \conch{} $p$-values $\{p_{t,n}\}$ is that here, when evaluating $S_t^{(n)}(\pi(\bx))$, we \emph{do not} recompute $\hat{\xi}_n$ under the permutation. These `frozen-estimator' $p$-values $\{\tilde p_{t,n}\}$ will serve as intermediate quantities in our analysis.

\subsubsection{Consistency of the MLE estimate}\label{app:mle_consistency}
\begin{theorem}\label{thm:mle_oracle_consistency}
    Suppose, at sample size $n$, $\hat{\bxi}_{n}:=\hat{\bxi}_{n}(\bX)$ be as defined in~\eqref{eq:mle}, and that $0<\sigma_\star<\infty$. Then, it holds that
    \[\|\hat{\bxi}_{n}-\bxi_n\|_\infty=\mathrm{O}_P(\log K_n).\]
\end{theorem}
\begin{proof}
We start by fixing $k\in [K_n]$. Recall that
\[
\hat{\xi}_{k,n}=\hat{\xi}_{k,n}(\bX)\in \argmax_{s\in [n-1]} L_k(s),\qquad \text{where}~L_k(s)=\sum_{i=1}^s \ell_k(X_{i,k,n}).
\]
For any $t>\xi_{k,n}$, we can write $L_k(t)=L_k(\xi_{k,n})+\sum_{s=\xi_{k,n}+1}^t \ell_k(X_{s,k,n})$.
By a union bound, for any $M\in \N$, we obtain
\begin{align*}
    \P(\,\hat{\xi}_{k,n}\ge \xi_{k,n}+M)
    &= \P\bigl(\cup_{t\,\ge\,\xi_{k,n}+M}\{L_k(t)\ge L_k(\xi_{k,n})\}\bigr)\\
    &\le \sum_{t\,\ge\,\xi_{k,n}+M}^n \P\,(L_k(t)\ge L_k(\xi_{k,n}))= \sum_{t\,\ge\,\xi_{k,n}+M}^n \P\,\Bigl(\sum_{s=\xi_{k,n}+1}^t \ell_k(X_{s,k,n})\ge 0\Bigr).
\end{align*}
Observe that $\{X_{s,k,n}\}_{s=\xi_{k,n}+1}^t$ are i.i.d samples from $P_{1,k}$. Therefore, by Lemma~\ref{lem:negative_drift_partial_sums}, there exists a constant $\gamma>0$ such that
\begin{align*}
    \P(\,\hat{\xi}_{k,n}\ge \xi_{k,n}+M)\le \sum_{t\,\ge\,\xi_{k,n}+M}^n e^{-\gamma\, (t-\xi_{k,n})}\le C_0\,C_1^M,
\end{align*}
for some $C_0>0$ and $C_1<1$. By another union bound, we have 
\begin{align*}
    \P(\,\text{there exists $k\in [K_n]$ such that}\,\, \hat{\xi}_{k,n}\ge \xi_{k,n}+M)\le K_n,C_0\,C_1^M,
\end{align*}
Similarly, we can show that there exists $C'_0>0$ and $C'_1<1$ such that 
\begin{align*}
    \P(\,\text{there exists $k\in [K]$ such that}\,\, \hat{\xi}_{k,n}\le \xi_{k,n}-M)\le K_n,C'_0\, {C'_1}^M,
\end{align*}
Together, this yields $\|\hat{\bxi}_{n}-\bxi_n\|_\infty=\mathrm{O}_P(\log K_n)$, as required.
\end{proof}

\subsection{Proving asymptotic sharpness of \croc}\label{app:uniform_decay_croc_pvalue}

Without loss of generality, we may assume $k_\star=1$. Recalling that for each $k\in [K_n]$, the \croc{} $p$-values are defined as $p_{(k),n}=\max_{\bt\in I_{k,n}}p_{\bt,n}$, it suffices to show that
\[
\max_{k\neq 1}\max_{\bt\in I_{k,n}}p_{\bt,n}\overset{P}{\longrightarrow}0.
\]
Throughout this proof, we write
\[
\Delta_n=\min_{j\in [K_n]: j\neq 1} (\xi_{j,n}-\xi_{1,n}),
\]
which denotes the separation between the changepoint of the root-cause stream and that of the remaining streams.

The proof from here is split into four key steps. 

\paragraph{Step 1: Reduction to a high probability event, and frozen-estimator p-values.}
First, we note that by Lemma~\ref{thm:mle_oracle_consistency}, there exists $C>0$ such that, writing $r_n:=C\log(K_n)$ and defining $\mathcal E_n:=\{\|\hat\bxi_n-\bxi_n\|_\infty\le r_n\}$, we have $\P(\mathcal E_n)\to 1$ as $n\to\infty$. Moreover, on the event $\mathcal E_n$, for every $k\neq 1$,
\[
\hat\xi_{k,n}-\hat\xi_{1,n} \ge \xi_{k,n}-\xi_{1,n}-2r_n \ge \Delta_n-2r_n.
\]
Now fix any $k\neq 1$ and any $\bt\in I_{k,n}$. Since by definition, $t_k<t_1$, by triangle inequality,
\[
|t_1-\hat\xi_{1,n}|+|t_k-\hat\xi_{k,n}|
\ge | (\hat\xi_{k,n}-\hat\xi_{1,n}) + (t_1-t_k)|
\ge \hat\xi_{k,n}-\hat\xi_{1,n}.
\]
Thus, it follows that
\[
|t_1-\hat\xi_{1,n}|+|t_k-\hat\xi_{k,n}|\ge \Delta_n-2r_n.
\]
Therefore, writing $M(\bt):=\sum_{j=1}^{K_n} |t_j-\hat\xi_{j,n}|$,
we have, uniformly over all wrong-root configurations,
\[
\inf_{k\neq 1}\inf_{\bt\in I_{k,n}}M(\bt)\ge \Delta_n-2r_n.
\]
Since by the theorem statement, $\Delta_n/r_n\to\infty$, on $\mathcal E_n$, for all sufficiently large $n$, $M(\bt)\ge \frac{\Delta_n}{2}$.

Let $\tilde p_{\bt,n}$ denote the frozen-estimator p-value, defined in~\eqref{eqn:pvalue_for_consistency}. By Lemma~\ref{lem:pvalue_inequality}, deterministically, $p_{\bt,n}\le \tilde p_{\bt,n}$ for all $\bt\in \Rcal$.
Therefore, it suffices to prove that
\[
\sup_{k\neq 1}\sup_{\bt\in I_{k,n}}\tilde p_{\bt,n}\overset{P}{\longrightarrow}0.
\]

\paragraph{Step 2: Deterministic upper bound on frozen-estimator $p$-values.}
First, note that $\mathrm{Var}_{X\sim P}(\ell_j(X))\le \sigma_\star^2$ and $\E_{X\sim P}(\ell_j(X))\le J_{\max}$  for both $P=P_{0,j}$ and $P=P_{1,j}$. Therefore, it follows that
\[
\E_{X\sim P_{0,j}} \ell_j^2(X),\quad \E_{X\sim P_{1,j}} \ell_j^2(X) \le \sigma_\star^2+J_{\max}^2.
\]
Furthermore, for any nonnegative weights $(a_1,\dots,a_{K_n})$ with $\sum_{j=1}^{K_n} a_j=1$,
\[
\E\Big[\sum_{j=1}^{K_n} a_j \hat v_{n,j}\Big] \le \sigma_\star^2+J_{\max}^2.
\]
Fix $\eta\in (0,1)$.
Applying Markov's inequality yields that, with probability at least $1-\eta$,
\[
\sum_{j=1}^{K_n} a_j \hat v_{n,j} \le \frac{\sigma_\star^2+J_{\max}^2}{\eta}.
\]
Choosing $a_j = \frac{d_j(\bt)}{M(\bt)}$ for $j\in [K_n]$, we obtain
\[
\sum_{j=1}^{K_n} d_j(\bt)\hat v_{n,j}
\le \frac{\sigma_\star^2+J_{\max}^2}{\eta} M(\bt).
\]
Since the changepoints $(\xi_{n,j})$ lie in the interior, we also have for each $j\in [K_n]$,
\[
\frac{n}{\min(\hat\xi_{n,j},n-\hat\xi_{n,j})} \le \frac{1}{\tau}.
\]
Hence, by Lemma~\ref{lem:deterministic_bound_onn_tildep}, we conclude that
\begin{equation}\label{eq:frozen_pvalue_event}
    \P\left(\sup_{k\neq 1}\sup_{\bt\in I_{k,n}}\tilde p_{\bt,n}
\le \sup_{k\neq 1}\sup_{\bt\in I_{k,n}}\left\{
\frac{
\left(\frac{\sigma_\star^2+J_{\max}^2}{\eta\tau}\right)M(\bt)
}{\{S^{(n)}(X,\bt)-\Gamma(\bt,\hat\bxi_n)\}^2}
+\mathbf 1\{S^{(n)}(X,\bt)-\Gamma(\bt,\hat\bxi_n)\ge0\}\right\}\right)\ge 1-\eta.
\end{equation}

Next, we prove a uniform negative drift bound for $S^{(n)}(X,\bt)-\Gamma(\bt,\hat\bxi_n)$ over all wrong-root configurations.

\paragraph{Step 3: Uniform negative drift of $S^{(n)}(X,\bt)-\Gamma(\bt,\hat\bxi_n)$.}
We now establish a uniform negative drift for $S^{(n)}(\bX,\bt)-\Gamma(\bt,\hat\bxi_n)$. In particular, we show that there exists $c_0>0$ such that
\begin{equation}\label{eq:negative_drift_event}
    \P\!\left(\left\{\sup_{k\neq 1}\sup_{\bt\in I_{k,n}}
\frac{S^{(n)}(X,\bt)-\Gamma(\bt,\hat\bxi_n)}{M(\bt)}\le -c_0\right\}\cap \mathcal{E}_n\right)\ge 1-\eta +\mathrm{o}(1).
\end{equation}
For each $j\in[K_n]$, we define $d_j(\bt):=|t_j-\hat\xi_{j,n}|$ and note that 
\[
S^{(n)}(\bX,\bt)-\Gamma(\bt,\hat\bxi_n)=\sum_{j=1}^{K_n}D_j(t_j,\hat\xi_{j,n}),
\]
where we let
\[
D_j(t_j,\hat\xi_{j,n}):=
\begin{cases}
\sum_{i=\hat\xi_{j,n}+1}^{t_j}\ell_j(X_{i,j,n})-\frac{d_j(\bt)}{t_j}\sum_{i=1}^{t_j}\ell_j(X_{i,j,n}),& t_j>\hat\xi_{j,n},\\[2ex]
\sum_{i=t_j+1}^{\hat\xi_{j,n}}(\ell_j(X_{i,j,n}))+\frac{d_j(\bt)}{n-t_j}\sum_{i=t_j+1}^{n}\ell_j(X_{i,j,n}),& t_j<\hat\xi_{j,n},\\[2ex]
0,& t_j=\hat\xi_{j,n}.
\end{cases}
\]
Now, for each $j\in [K_n]$, we derive a upper bound on each $D_j(t_j,\hat\xi_{j,n})$ that holds uniformly over all choices of $t_j\in [n-1]$. Consider first the case $t_j>\hat\xi_{j,n}$. On the event $\mathcal E_n$, a direct decomposition gives $D_j(t_j,\hat\xi_{j,n})\le T_{1j}+T_{2j}+T_{3j}$,
where we define
\begin{align*}
T_{1j}
:=
-\frac{d_j(\bt)}{t_j}
&\sum_{i=1}^{\xi_{j,n}-r_n}\ell_j(X_{i,j,n}),\qquad T_{2j}:=\sum_{|i-\xi_{j,n}|\le r_n}|\ell_j(X_{i,j,n})|,\\
&T_{3j}:=\frac{\hat\xi_{j,n}}{t_j}\sum_{i=\xi_{j,n}+r_n+1}^{t_j}\ell_j(X_{i,j,n}).
\end{align*}
Here $T_{3j}$ is interpreted as zero if $\xi_{j,n}+r_n+1>t_j$. We control each of these terms one by one. First, since
$X_{1,j,n},\ldots,X_{\xi_{j,n}-r_n,j,n}$ are i.i.d.\ from $P_{0,j}$,
Lemma~\ref{lem:negative_drift_partial_sums} gives that, for each $j$,
\[
\P\!\left(
\frac{1}{\xi_{n,j}-r_n}
\sum_{i=1}^{\xi_{n,j}-r_n}\ell_j(X_{i,j})
\ge cJ_{\min}
\right)
\ge 1-e^{-\gamma(\xi_{n,j}-r_n)}.
\]
We now use the interiority condition to lower bound the ratio $(\xi_{j,n}-r_n)/t_j$. Since $t_j\le n$ and $\xi_{j,n}\ge\tau n$ for all $j\in[K_n]$,
\[
\frac{\xi_{j,n}-r_n}{t_j}
\ge
\frac{\xi_{j,n}-r_n}{n}
\ge
\tau-\frac{r_n}{n}.
\]
As $r_n=\mathrm{O}(\log K_n)=\mathrm{o}(n)$, for all sufficiently large $n$, uniformly over $j$, $\frac{\xi_{j,n}-r_n}{t_j}\ge \frac{\tau}{2}$.
Consequently, setting $c_1:=\frac{c\tau J_{\min}}{2}$, we have
\[
\P\left(T_{1j}\le -c_1d_j(\bt)~~\textnormal{for all}~~t_j\in [n-1]\right)\ge 1-e^{-\gamma(\xi_{j,n}-r_n)}
\] 
Next, we control $T_{3j}$ uniformly over $t_j$. Since, in the particular case of consideration, $\hat\xi_{j,n}/t_j\le1$, we further have
\[
T_{3j}
\le
\sup_{s\ge \xi_{j,n}+r_n+1}
\sum_{i=\xi_{j,n}+r_n+1}^{s}\ell_j(X_{i,j,n}).
\]
For $s>\xi_{j,n}+r_n$, we observe that $X_{\xi_{j,n}+r_n+1,j,n},\ldots, X_{s,j,n}$ are drawn i.i.d.\ from the distribution $P_{1,j}$.
By Lemma~\ref{lem:negative_drift_partial_sums}, we obtain that there exists $c'>0$ and $\gamma'>0$
\[
\P\!\left(
\sup_{s\ge \xi_{j,n}+r_n+1}
\sum_{i=\xi_{j,n}+r_n+1}^{s}\ell_j(X_{i,j,n})
>c'r_n
\right)
\le e^{-\gamma'r_n}.
\]
A union bound over $j\in[K_n]$ gives
\[
\P\!\left(
\sup_{j\in[K_n]}
\sup_{s\ge \xi_{j,n}+r_n+1}
\sum_{i=\xi_{j,n}+r_n+1}^{s}\ell_j(X_{i,j,n})
>c'r_n
\right)
\le K_ne^{-\gamma'r_n}.
\]

If instead, $t_j<\hat\xi_{j,n}$ we follow an analogous argument to bound $D_j(t_j,\hat\xi_{j,n})$. Specifically, on $\mathcal E_n$, we write $D_j(t_j,\hat\xi_{j,n})\le T'_{1j}+T_{2j}+T'_{3j}$,
where $T_{2j}$ is unchanged and
\[
T'_{1j}:=
\frac{d_j(\bt)}{n-t_j}
\sum_{i=\xi_{j,n}+r_n+1}^{n}\ell_j(X_{i,j,n}),
\qquad
T'_{3j}:=
-\frac{n-\hat\xi_{j,n}}{n-t_j}
\sum_{i=t_j+1}^{\xi_{j,n}-r_n}\ell_j(X_{i,j,n}).
\]
We can follow the argument from above and derive similar upper bounds on $T'_{1j}$ uniformly for all $t_j\in [n-1]$, and on $T'_{3j}$ uniformly for all $t_j\in [n-1]$ and $j\in [K_n]$. Combiningn both cases, we get the following: define the event
\[
\mathcal A_n:=
\left\{D_j(t_j,\hat\xi_{j,n})
\le-c_1d_j(\bt)+T_{2j}+c'r_n\ \text{for all }j\in[K_n]\text{ and all }t_j\in[n-1]
\right\}\cap \mathcal{E}_n.
\]
The preceding bounds imply
\[
\P(\mathcal A_n^c)
\le\sum_{j=1}^{K_n}\left\{
e^{-\gamma(\xi_{j,n}-r_n)}+e^{-\gamma(n-\xi_{j,n}-r_n)}
\right\}+2K_ne^{-\gamma_3r_n}.
\]
By the interiority condition, we have
\[
\xi_{j,n}\wedge(n-\xi_{j,n})\ge \frac{\tau n}{2}
\qquad\text{for all }j\in[K_n].
\]
Therefore,
\[
\sum_{j=1}^{K_n}
\left\{e^{-\gamma(\xi_{j,n}-r_n)}
+e^{-\gamma(n-\xi_{j,n}-r_n)}
\right\}\le
2K_ne^{-\gamma(\tau n/2-r_n)}=o(1).
\]
Also, since $r_n=C\log K_n$, choosing $C$ large enough, $2K_ne^{-\gamma'r_n}=o(1)$. Thus, $\P(\mathcal A_n)\to1$.
On $\mathcal A_n$, summing over $j\in [K_n]$ gives, uniformly over $\bt\in\Rcal$,
\[
S^{(n)}(X,\bt)-\Gamma(\bt,\hat\bxi_n)
\le
-c_1M(\bt)+\sum_{j=1}^{K_n}T_{2j}+c'K_nr_n.
\]
It remains to control the middle term. Define
\[
\mathcal B_n(\eta):=
\left\{
\sum_{j=1}^{K_n}T_{2j}
\le
\frac{2K_nr_n(\sigma_\star^2+J_{\max}^2)^{1/2}}{\eta}
\right\}.
\]
By Markov's inequality and Cauchy--Schwarz, $\P(\mathcal B_n(\eta))\ge1-\eta.$
Consequently, on $\mathcal A_n\cap\mathcal B_n(\eta)$, uniformly over $\bt\in\Rcal$,
\[
S^{(n)}(X,\bt)-\Gamma(\bt,\hat\bxi_n)\le-c_1M(\bt)+C_\eta K_nr_n,
\]
where we let
\[
C_\eta:=c'+\frac{2(\sigma_\star^2+J_{\max}^2)^{1/2}}{\eta}.
\]

Finally, in Step~1, we proved that on the event $\mathcal E_n$,
\[
\inf_{k\neq1}\inf_{\bt\in I_{k,n}}M(\bt)\ge\Delta_n-2r_n.
\]
Therefore, on $\mathcal A_n\cap\mathcal B_n(\eta)$,
\[
\sup_{k\neq1}\sup_{\bt\in I_{k,n}}
\frac{S^{(n)}(X,\bt)-\Gamma(\bt,\hat\bxi_n)}{M(\bt)}
\le
-c_1+\frac{C_\eta K_nr_n}{\Delta_n-2r_n}.
\]
Since $r_n=C\log K_n$ and $\Delta_n\gg K_n\log K_n$, the second term is $o(1)$. Hence, for all sufficiently large $n$, on $\mathcal A_n\cap\mathcal B_n(\eta)$,
\[
\sup_{k\neq1}\sup_{\bt\in I_{k,n}}
\frac{S^{(n)}(X,\bt)-\Gamma(\bt,\hat\bxi_n)}{M(\bt)}
\le -\frac{c_1}{2}.
\]
Therefore,~\eqref{eq:negative_drift_event} holds with $c_0=c_1/2$.

\paragraph{Step~4: Completing the proof.}
By the conclusion of Step~2, and Step~3, that is~\eqref{eq:frozen_pvalue_event} and \eqref{eq:negative_drift_event}, we have that except on an event with probability at most $1-2\eta$, for every
$\bt\in\bigcup_{k\neq1}I_k$, \[
\tilde p_{\bt,n}
\le
\frac{\left(\frac{\sigma_\star^2+J_{\max}^2}{\eta\tau}\right)M(\bt)
}{c_0^2M(\bt)^2}=
\frac{\sigma_\star^2+J_{\max}^2}{\eta\tau c_0^2}\cdot \frac{1}{M(\bt)},
\]
and also that for all sufficiently large $n$,
\[
M(\bt)\ge\frac{\Delta_n}{2}.
\]
Hence, on this event,
\[
\sup_{k\neq1}\sup_{\bt\in I_k}\tilde p_{\bt,n}
\le\frac{2(\sigma_\star^2+J_{\max}^2)}
{\eta\tau c_0^2\Delta_n}\to0
\]
Since $\eta>0$ is arbitrary to start with,
\[
\max_{k\neq1}\max_{\bt\in I_k}\tilde p_{\bt,n}\overset{P}{\longrightarrow}0.
\]
This proves the first part. Finally, for any fixed $\alpha\in(0,1)$, note that
\[
\P\left(\exists k\neq1:k\in\mathcal K^{\croc}_{n,1-\alpha}\right)
=
\P\left(\max_{k\neq1}p_{(k),n}>\alpha\right)
\to0.
\]
Thus, $\P\left(\mathcal K^{\croc}_{n,1-\alpha}\subseteq\{1\}\right)\to1$.
By Theorem~\ref{thm:validity_croc}, we further have
\[
\P\left(1\in\mathcal K^{\croc}_{n,1-\alpha}\right)\ge1-\alpha.
\]
Therefore,
\[
\P\left(\mathcal K^{\croc}_{n,1-\alpha}=\{1\}\right)\ge1-\alpha-o(1).
\]
This proves the second part.
\subsubsection{Auxiliary lemmas}\label{app:lemma_consistency}

\begin{lemma}\label{lem:negative_drift_partial_sums}
Fix $k\in\mathbb N$. Suppose $X_1,\ldots,X_k\stackrel{\mathrm{iid}}{\sim}P_1$ and
$Y_1,\ldots,Y_k\stackrel{\mathrm{iid}}{\sim}P_0$, such that $P_0$ and $P_1$ admit densities $f_0$ and $f_1$ w.r.t. a common dominating measure. We write $\ell(x)=\log(f_0(x)/f_1(x))$ and further, assume that
$0<\mathrm{KL}(P_0\|P_1),\ \mathrm{KL}(P_1\|P_0)<\infty$.
Then there exist constants $c_0>0$ and $\gamma_0>0$ such that
\[
\P\biggl(\frac{1}{k}\sum_{i=1}^k \ell(X_i)\ge -c_0J(P_0,P_1)\biggr)\le e^{-\gamma_0 k},
\qquad
\P\biggl(\frac{1}{k}\sum_{i=1}^k \ell(Y_i)\le c_0J(P_0,P_1)\biggr)\le e^{-\gamma_0 k}.
\]
Moreover, there exist constants $c_1,\gamma_1>0$ such that, for every $x>0$,
\[
\P\left(\sup_{m\ge 1}\sum_{i=1}^m \ell(X_i)>x\right)\le c_1e^{-\gamma_1 x},
\qquad
\P\left(\sup_{m\ge 1}\sum_{i=1}^m -\ell(Y_i)>x\right)\le c_1e^{-\gamma_1 x}.
\]
\end{lemma}

\begin{proof}
For a proof of the first part, we refer the reader to Lemma $B.4$ in \cite{hore2026conformal}.
For the second part, observe that for any $\theta\in(0,1)$, by H\"older's inequality,
\[
M(\theta):=\E_{X\sim P_1}[e^{\theta\ell(X)}]
=
\int f_0^\theta(x)f_1^{1-\theta}(x)\,\mathsf dx
\le
\left(\int f_0(x)\,\mathsf dx\right)^\theta
\left(\int f_1(x)\,\mathsf dx\right)^{1-\theta}
\le 1.
\]
Moreover, $M(0)=1$ and
\[
\lim_{\theta\to0^+}M'(\theta)=\E_{X\sim P_1}[\ell(X)]=-\mathrm{KL}(P_1\|P_0)<0.
\]
Hence there exists $\theta_\star\in(0,1)$ such that
\[
\rho:=\E_{X\sim P_1}[e^{\theta_\star\ell(X)}]=M(\theta_\star)<1.
\]
Next, note that for any $x>0$, by a union bound and Chernoff's inequality,
\begin{align*}
\P\left(\sup_{m\ge1}\sum_{i=1}^m\ell(X_i)>x\right)
&\le
\sum_{m=1}^\infty
\P\left(\sum_{i=1}^m\ell(X_i)>x\right)\\
&\le
\sum_{m=1}^\infty
e^{-\theta_\star x}\,\E\left[\exp\left(\theta_\star\sum_{i=1}^m\ell(X_i)\right)\right]\\
&=e^{-\theta_\star x}\sum_{m=1}^\infty \rho^m=
\frac{\rho}{1-\rho}e^{-\theta_\star x}.
\end{align*}
Thus the first maximal bound holds with
\[
c_1=\frac{\rho}{1-\rho},\qquad \gamma_1=\theta_\star.
\]
The second maximal bound follows by a similar argument.

\end{proof}

\begin{lemma}\label{lem:deterministic_bound_onn_tildep}
For any $\bt\in \Rcal$, we have deterministically
\begin{equation}\label{eq:deterministic_bound_onn_tildep}
\tilde{p}_{\bt,n}\le \frac{\sum_{k=1}^K \left(\frac{n}{\min(\hat{\xi}_{k,n},n-\hat{\xi}_{k,n})} \cdot |t_k-\hat\xi_{k,n}|\cdot  \hat{v}_{n,k}\right)}{(S^{(n)}(\bX,\bt)-\Gamma(\bt,\hat\bxi_n))^2} + \one{S^{(n)}(\bX,\bt)-\Gamma(\bt,\bxi_n)\ge 0}.
\end{equation}
\end{lemma}

\begin{proof}
Fix $\bt\in \Rcal$ and write $d_{k}(\bt)=|t_k-\hat\xi_{k,n}|$ We start by writing
\[
S^{(n)}(\bx,\bt)=\sum_{k=1}^K {\rm sgn}(t_k-\hat{\xi}_{k,n})\cdot\biggl\{\sum_{i=t_k\wedge \hat{\xi}_{k,n}(\bx)+1}^{t_k\vee \hat{\xi}_{k,n}}\ell_k(x_{i,k})\biggr\},
\]
i.e., the inner sum is a sum of $\ell_k$ computed over a block of $d_k(\bt)$ observations. Thereby, the permutation $p$-value \eqref{eqn:pvalue_for_consistency} can be equivalently expressed as
\[
\tilde{p}_{\bt,n}:=\P_{\pi\sim \mathrm{Unif}(\Pi_{\bt})}\!\left(S^{(n)}(\pi(\bX),\bt)\le S^{(n)}(\bX,\bt)\ \middle|\ \bX\right).
\]
Let for each $k\in [K]$, 
\[
\mathcal{X}_k=\begin{cases}
    \{X_{1,k,n},\ldots,X_{t_k,k,n}\} & \text{if}~~t_k>\hat\xi_{k,n},\\
    \{X_{t_k+1,k,n},\ldots,X_{n,k,n}\} & \text{if}~~t_k<\hat\xi_{k,n},
\end{cases}.
\]
and suppose for each $k\in [K]$, $\tilde{X}_{1,k,n},\ldots,\tilde{X}_{d_k(\bt),k,n}$ are drawn WOR from $\mathcal{X}_k$. Therefore, we can equivalently express $\tilde{p}_{\bt,n}$ as
\[
\tilde{p}_{k,n}
=\P\!\left(\sum_{k=1}^K\sum_{j=1}^{d_k(\bt)}\ell_k(\tilde{X}_{j,k,n})\le S^{(n)}(\bX,\bt)\ \middle|\ \bX\right).
\]
Furthermore, observe that
\[
\E\left[\sum_{k=1}^K\sum_{j=1}^{d_k(\bt)}\ell_k(\tilde{X}_{j,k,n})\ \middle| \ \bX\right]=\Gamma(\bt,\hat{\bxi}_n)
\]
and that
\[
\mathrm{Var}\!\left(\sum_{k=1}^K\sum_{j=1}^{d_k(\bt)}\ell_k(\tilde{X}_{j,k,n})\ \middle| \ \bX\right)=\sum_{k=1}^K\mathrm{Var}\!\left(\sum_{j=1}^{d_k(\bt)}\ell_k(\tilde{X}_{j,k,n})\ \middle| \ \bX\right).
\]
Moreover, 
\begin{align*}
    \mathrm{Var}\!\left(\sum_{j=1}^{d_k(\bt)}\ell_k(\tilde{X}_{j,k,n})\ \middle| \ \bX\right)&\le d_k(\bt)\cdot \mathrm{Var}(\ell_k(\tilde{X}_{j,k,n})\ | \ \bX)\\
    &\le  d_k(\bt)\cdot \begin{cases}
    \frac{1}{t_k}\sum_{i=1}^{t_k} \ell^2_{k}(X_{i,k,n}) & \text{if } t_k>\hat{\xi}_{k,n}\\
    \frac{1}{n-t_k}\sum_{i=t_k+1}^{n} \ell^2_{k}(X_{i,k,n}) & \text{if } t_k<\hat{\xi}_{k,n}
    \end{cases}\\
    &\le d_k(\bt)\cdot \frac{n}{\min(\hat{\xi}_{k,n},n-\hat{\xi}_{k,n})} \, \hat{v}_{n,k}
\end{align*}
where recall $\hat{v}_{n,k}=\frac{1}{n}\sum_{i=1}^{n} {\ell_k}^{\,2}(X_{i,k,n})$.
Hence, writing
\begin{multline*}
\P\!\left(\sum_{k=1}^K\sum_{j=1}^{d_k(\bt)}\ell_k(\tilde{X}_{j,k,n})\le S^{(n)}(\bX,\bt)\ \middle|\ \bX\right)\\
    \le \P\!\left(\sum_{k=1}^K\sum_{j=1}^{d_k(\bt)}\ell_k(\tilde{X}_{j,k,n})-\Gamma(\bt,\hat{\bxi}_n)\le S^{(n)}(\bX,\bt)-\Gamma(\bt,\hat{\bxi}_n)\ \middle|\ \bX\right).
\end{multline*}
Observe that whenever $S^{(n)}(\bX,\bt)-\Gamma(\bt,\hat{\bxi}_n)<0$, then by Chebyshev's inequality,
\[
\tilde{p}_{\bt,n} \le \,\frac{\sum_{k=1}^K \frac{n}{\min(\hat{\xi}_{n,k},n-\hat{\xi}_{n,k})} \cdot d_k(\bt) \hat{v}_{n,k}}{(S^{(n)}(\bX,\bt)-\Gamma(\bt,\hat\bxi_n))^2}.
\]
On the other hand, if $S^{(n)}(\bX,\bt)-\Gamma(\bt,\hat\bxi_n)\ge 0$, then $\tilde{p}_{\bt,n}\le 1$. This completes the proof. 
\end{proof}

\begin{lemma}\label{lem:pvalue_inequality}
    Let $\{p_{\bt,n}\}$ be the \croc{} $p$-values computed based on the oracle CPP score. Then, it holds that deterministically,
    \[
    p_{t,n}\ \le\ \tilde{p}_{t,n}, \qquad \text{for}~~t\in \Rcal.
    \]
\end{lemma}

\begin{proof}
    We start with observing that the original score $S(\bX,\bt)$  and the frozen score $S^{(n)}(\bX,\bt)$ can be equivalently written as 
    \begin{align*}
        S(\bx,\bt)\ &= \ \sum_{k=1}^K\left(\sum_{i=1}^{t_k} \ell_{k}(x_{i,k,n})-\sum_{i=1}^{\hat{\xi}_{k,n}(\bx)} \ell_{k}(x_{i,k,n})\right), \\ S^{(n)}(\bx,\bt)\ &= \ \sum_{k=1}^K\left(\sum_{i=1}^{t_k} \ell_{k}(x_{i,k,n})-\sum_{i=1}^{\hat{\xi}_{k,n}} \ell_{k}(x_{i,k,n})\right),
    \end{align*}
    where $(\hat{\xi}_{1,n},\ldots,\hat{\xi}_{K,n})=\hat{\bxi}_n=\hat{\bxi}_n(\bX)$.
    Recall that by \eqref{eq:mle}, for any $\bx\in \prod_{k=1}^K\mathcal{X}_k^n$,
    \[
    \hat{\xi}_{k,n}(\bx)\in \argmax_{s\in [n-1]} \sum_{i=1}^s \ell_{k}(x_{i,k,n}).
    \]
    Hence, for each $k\in [K]$, and for any permutation $\pi\in \mathcal{S}_n$, we have that 
    \[
    \sum_{i=1}^{\hat{\xi}_{k,n}(\pi(\bx))} \ell_{k}((\pi(\bx))_{i,k,n})\ \ge \sum_{i=1}^{\hat{\xi}_{k,n}(\bx)} \ell_{k}((\pi(\bx))_{i,k,n}).
    \]
    Consequently, for any $t\in \Rcal$ and any permutation $\pi\in \Pi_t$,
    \[
    S(\pi(\bX),\bt)\leq S^{(n)}(\pi(\bX),\bt),\qquad S(\bX,\bt)=S^{(n)}(\bX,\bt).
    \]
    Hence, for all $t\in \Rcal$, deterministically, $p_{\bt,n}\leq \tilde{p}_{\bt,n}$ as required.
\end{proof}

\end{document}